\documentclass[11pt,a4paper]{article}
\usepackage{amsmath}
\usepackage[pagebackref=true]{hyperref}
\usepackage[capitalize]{cleveref}
\usepackage{times}  
\usepackage{helvet}  
\usepackage{courier}  
\usepackage{graphicx} 
\usepackage[sort]{natbib}  
\usepackage{caption} 
\DeclareCaptionStyle{ruled}{labelfont=normalfont,labelsep=colon,strut=off} 
\usepackage{tabularx,booktabs}
\usepackage{amsmath,amssymb,graphicx,amsthm,dsfont}
\usepackage{todonotes}
\usepackage{xspace}
\usepackage[capitalize]{cleveref}
\usepackage{pgfplots}
\usepackage{thmtools} 
\usepackage{thm-restate}

\usepackage{nicefrac}

\usepackage[ruled,vlined,linesnumbered]{algorithm2e}
\DontPrintSemicolon
\SetKwInOut{Input}{Input}
\SetKwInOut{Output}{Output}

\crefalias{AlgoLine}{line}

\crefname{algocf}{Algorithm}{Algorithms}

\usepackage{etoolbox}
\makeatletter
\patchcmd{\@algocf@start}
  {-1.5em}
  {0pt}
  {}{}
\makeatother

\makeatletter
\let\cref@old@stepcounter\stepcounter
\def\stepcounter#1{%
  \cref@old@stepcounter{#1}%
  \cref@constructprefix{#1}{\cref@result}%
  \@ifundefined{cref@#1@alias}%
    {\def\@tempa{#1}}%
    {\def\@tempa{\csname cref@#1@alias\endcsname}}%
  \protected@edef\cref@currentlabel{%
    [\@tempa][\arabic{#1}][\cref@result]%
    \csname p@#1\endcsname\csname the#1\endcsname}}
\makeatother

\newcommand{\mytodo}[2]{\xspace}
\newcommand{\myrevtodo}[2]{{%
		\let\marginpamarginnote
		\reversemarginpar
		\renewcommand{\baselinestretch}{0.8}%
		}}
\newcommand{\myinlinetodo}[2]{\todo[size=\small, color=#1!50!white, inline,
	caption={}]{#2}\xspace}
\newcommand{\registerAuthor}[3]{%
	\expandafter\newcommand\csname #2com\endcsname[1]{\mytodo{#3}{\textsc{#2}:
			##1}}%
	\expandafter\newcommand\csname
	#2revcom\endcsname[1]{\myrevtodo{#3}{\textsc{#2}: ##1}}%
	\expandafter\newcommand\csname
	#2inline\endcsname[1]{\myinlinetodo{#3}{\textsc{#2}: ##1}}%
	\expandafter\newcommand\csname
	#2inlineLater\endcsname[1]{\lv{\myinlinetodo{#3}{\textsc{#2}: ##1}}}%
}
\registerAuthor{Niclas Boehmer}{nb}{orange}
\registerAuthor{André Nichterlein}{an}{green}
\registerAuthor{Robert Bredereck}{rb}{gray}

\usepackage{newfloat}
\usepackage{listings}
\newtheorem{corollary}{Corollary}
\newtheorem{proposition}{Proposition}

\newtheorem{lemma}{Lemma}

\newcommand{\decprob}[3]{%
  \begin{center}%
    \begin{minipage}{0.9\linewidth}%
      \textsc{#1}\\
      \textbf{Input:} #2\\
      \textbf{Question:} #3
    \end{minipage}%
  \end{center}%
}
\newcommand{\cfpr}{\textsc{Cycle-Free Reviewing}\xspace}
\newcommand{\wcfpr}{\textsc{Weighted Cycle-Free Reviewing}\xspace}
\newcommand{\otocfpr}{\textsc{Peer Cycle-Free Reviewing}\xspace}

\makeatletter
\providecommand*{\cupdot}{%
  \mathbin{%
    \mathpalette\@cupdot{}%
  }%
}
\newcommand*{\@cupdot}[2]{%
  \ooalign{%
    $\m@th#1\cup$\cr
    \hidewidth$\m@th#1\cdot$\hidewidth
  }%
}
\makeatother

\newcommand{\ddp}{\ensuremath{d_{\text{paper}}}}
\newcommand{\ccr}{\ensuremath{c_{\text{reviewer}}}}
\newcommand{\N}{\ensuremath{\mathds{N}}}

\DeclareMathOperator{\aut}{\text{aut}}

\DeclareMathOperator{\rev}{\text{rev}}

\DeclareMathOperator{\coi}{coi}

\usepackage{authblk}
\usepackage{fullpage}
\title{Combating Collusion Rings is Hard but Possible}
\author[1]{Niclas Boehmer}
\author[2]{Robert Bredereck}
\author[1]{Andr{\'{e}} Nichterlein}
\affil[1]{ \small Technische Universit{\"a}t Berlin,  Faculty IV, Algorithmics and Computational Complexity, Berlin, Germany\protect\\
\{niclas.boehmer,andre.nichterlein\}@tu-berlin.de}
\affil[2]{ \small Humboldt-Universit{\"a}t zu Berlin, Institut f{\"u}r Informatik, Algorithm Engineering, Berlin, Germany\protect\\ robert.bredereck@hu-berlin.de}

\begin{document}

\maketitle

\begin{abstract}
A recent report of Littmann [Commun. ACM '21] outlines the existence and the fatal impact of collusion rings in academic peer reviewing.
We introduce and analyze the problem \cfpr that aims at finding a review assignment without the following kind of collusion ring:
A sequence of reviewers each reviewing a paper authored by the next reviewer in the sequence (with the last reviewer reviewing a paper of the first),
thus creating a \emph{review cycle} where each reviewer gives favorable reviews.
As a result, all papers in that cycle have a high chance of acceptance independent of their respective scientific merit.

We observe that review assignments computed using a standard Linear Programming approach typically admit many short review cycles. 
On the negative side, we show that \cfpr is NP-hard in various restricted cases
(i.e., when every author is qualified to review all papers and one wants to prevent that authors review each other's or their own papers
or when every author has only one paper and is only qualified to review few papers).
On the positive side, among others, we show that, in some realistic settings, an assignment without any review cycles of small length always exists. This result also gives rise to an efficient heuristic for computing (weighted) cycle-free review assignments, which we show to be of excellent quality in practice.  
\end{abstract}

\section{Introduction}
As recently pointed out by \citet{DBLP:journals/cacm/Littman21},
the integrity and legitimacy of scientific conference publications (particularly important in the context of computer science)
is threatened by so-called ``collusion rings'', which are sets of authors that unethically review and support
each other while breaking anonymity and hiding conflicts of interest.
Despite the fact that details are usually not disclosed for various reasons, it is inevitable that the process of assigning papers to reviewers
is the key point to engineer technical barriers against such incidents.
Whereas assignments at very small venues could be performed manually, support by (semi-)automatic systems becomes necessary
already for medium-size conferences.
Today computational support for finding review assignments is well-established and has improved the quality of the reviewing and paper assignment
process in many ways (see the surveys of \citet{survey} and \citet{DBLP:journals/cacm/PriceF17} for details).
Still there is huge potential for improving processes
and further computational support is urgently requested~\citep{DBLP:journals/cacm/PriceF17,survey}.

When aiming to prevent collusion rings, one of the most basic properties one can request from a review assignment
is that the assignment does not contain any \emph{review cycle} of length~$z$, that is, a sequence of $z$~agents
each reviewing a paper authored by the next agent in the sequence (with the last agent reviewing a paper authored by the first).
This property is of high practical relevance:
For example, in the AAAI'21 review assignment the non-existence of review cycles of length at most~$z=2$
was a soft constraint~\citep{AAAI21intro}.
Yet, there is a lack of systematic studies concerning the computation of such assignments.
Motivated by this, we propose and analyze \cfpr, the problem of computing an assignment of papers to agents that is
free of review cycles of length at most~$z$, both from a theoretical and practical perspective.

\subsection{Related Work}
The literature is rich in the general context of peer reviewing (see, e.\,g., the works of \citet{DBLP:conf/aiconf,taylor2008optimal,DBLP:journals/algorithmica/GargKKMM10,DBLP:conf/icdm/LongWPY13,DBLP:conf/aaai/LianMNW18,DBLP:conf/kdd/KobrenSM19,JMLR:v22:20-190} on computational aspects of finding a ``good'' review assignment, and the survey of~\citet{survey}).
Closest to our work are \citet{DBLP:conf/atal/BarrotLPS20} and \citet{DBLP:conf/iccnc/Guo0CWL18}.
In the context of product reviewing, among others, \citet{DBLP:conf/atal/BarrotLPS20} propose and analyze a restricted case which translates to our setting as follows:
Given a set of single-author papers and a set of agents each writing a single paper and each having some conflicts of interest over papers, find a review assignment of papers
to agents, where each agent serves as a reviewer
providing one review and each paper must receive one review.
They show that in this setting finding an assignment without review cycles of length at most~$z$
corresponds to finding a $2$-factor without cycles of length at most~$z$, which is known to be NP-hard for~$z\ge 5$ but polynomial-time solvable for $z\leq 3$ \citep{DBLP:journals/siamdm/HellKKK88}. 
Closer to our setting is that of \citet{DBLP:conf/iccnc/Guo0CWL18}, who
also consider the computation of cycle-free review assignments.
They propose two simple heuristics
and conduct experiments measuring the quality of their heuristics and the number of review cycles in a weight-maximizing solution
on two instances, mostly focusing on the influence of the number of reviews per paper and per reviewer.

\subsection{Outline and Contributions}
Our contribution is threefold.
First, in \cref{sec:hardness}, we show the intractability of \cfpr in various restricted settings:
We show NP-hardness even when just forbidding review cycles of length at most two in ``sparse'' and ``dense'' settings (e.g., if each reviewer can review only ``few'' or can review ``almost all'' papers, see \cref{th:cfpr-sparse,th:full,th:cof}).
Furthermore, solving a question left open by \citet{DBLP:conf/atal/BarrotLPS20}, we show NP-hardness if each agent writes just one single-author paper and can review only few papers (\cref{th:otocfprNP}).

Second, in \cref{sec:heuristic}, we develop greedy heuristics.
In contrast to \citet{DBLP:conf/iccnc/Guo0CWL18} we provide a theoretical analysis for the heuristics.
In particular, we prove that, if the considered instance satisfies certain near-realistic conditions (such as that each paper has few authors and that for each paper there are many possible reviewers), then these heuristics are guaranteed to output a $z$-cycle-free review assignments in polynomial time. 

Third, in \cref{se:experiments}, we present and discuss the results of our experiments.
Our core results are:
\begin{enumerate}
	\item Existing linear-programming-based methods for computing maximum-weight review assignments (as often used in practice) produce assignments where a high fraction (20\% or more) of agents and papers belong to some review cycles of length two.
	\item For~$z \in \{2,3,4\}$ maximum-weight $z$-cycle-free assignments computed by one of our heuristics (see \cref{sec:heuristic}) or computed via Integer Linear Programming are almost as good as the maximum-weight review assignments with cycles (solution quality loss less than 4\% resp. 1\%).
	\item Somewhat surprisingly, we show that adding additional reviewers that are authors of some papers to the reviewer pool increases the number of papers that belong to
    review cycles in maximum-weight (non cycle-free) assignments.
\end{enumerate}

\section{Preliminaries}
\label{sec:preliminaries}

For~$n\in \N$, we set~$[n] := \{1,\ldots,n\}$.
In an instance of \cfpr, we are given a set $P$ of papers and set $A$ of agents, where each paper $p\in P$ is authored by a subset $\aut(p)\subseteq A$ of agents. 
Moreover, we are given for each agent $a\in A$ a subset $\rev(a)\subseteq P$ of papers the agent is \emph{qualified to review}\footnote{Being ``qualified to review'' can encode that the agent is capable of reviewing the paper or that the agent does not have a conflict of interest with one of the co-authors or both.}.
We capture this information in a bipartite graph $(A\cupdot P,E_A\cupdot E_P)$ with $E_A=\{(a,p)\mid a\in A, p\in \rev(a)\}$ and $E_P=\{(p,a)\mid p\in P, a\in 
\aut(p)\}$ (see also \cref{tab:notation} for an overview). 
A \emph{(peer) review assignment} $E'\subseteq E_A$ is a subset of edges from agents to papers, where we say that $a$ \emph{reviews} $p$ in $E'$ if $(a,p)\in E'$. 
Given a review assignment $E'\subseteq E_A$, for an agent $a\in A$, let $N^+(a,E')=\{p\in P\mid (a,p)\in E'\}$ be the subset of papers agent $a$ reviews in $E'$ and, for a paper $p\in P$, let $N^-(p,E')=\{a\in A\mid (a,p)\in E'\}$ be the subset of agents that review $p$ in $E'$. 
For~$c,d \in \N$ a review assignment $E'\subseteq E_A$ is called \emph{$c$-$d$-valid} if each agent reviews at most~$c$ papers and each paper is reviewed by~$d$ agents, that is, $|N^+(a,E')|\leq c$ for all $a\in A$ and $|N^-(p,E')|= d$ for all $p\in P$. 
In a review assignment $E'\subseteq E_A$, we say that papers $p_1,\dots, p_z$ and agents $a_1,\dots, a_z$ form a \emph{review cycle (of length $z$)} if $a_i$ is 
an author of $p_i$ \big($(p_i,a_i)\in E_P$\big) for all $i\in [z]$, $a_i$ reviews $p_{i+1}$ in $E'$ \big($(a_i,p_{i+1}) \in E'$\big) for $i\in [z-1]$ and $a_z$ reviews $p_1$ in $E'$ \big($(a_z,p_{1})  \in E'$\big).
Notably, a review cycle of length $z$ in $E'$ corresponds to a directed cycle of length $2z$ in $(A\cupdot P,E'\cupdot E_P)$ and a review cycle of length one corresponds to an author reviewing one of its own papers. We say that a review assignment $E'$ is $z$-cycle free if there is no review cycle of length $i\in [z]$ in $E'$. 

\begin{table}
	\caption{Notation overview}
	\label{tab:notation}%
	\setlength{\tabcolsep}{2pt}%
	\begin{tabularx}{\linewidth}{lX}
		\toprule
		Variable	& Explanation \\
		\midrule
		$V = A \cupdot P$ & vertex set consisting of agents~$A$ and papers~$P$ with~$n_A=|A|$ and~$n_P=|P|$\\
		$E_A$		& $(a,p) \in E_A \subseteq  A \times P$ shows~$a$ can review~$p$ \\
		$E_P$		& $(p,a) \in E_P \subseteq  P \times A$ shows~$a$ authors~$p$ \\
		$N^-(v,E)$	& in-neighbors of~$v \in V$ wrt.~$E \subseteq \binom{V}{2}$, i.\,e., $N^-(v,E) := \{u \in V \mid (u,v)\in E\}$ \\
		$N^+(v,E)$	& out-neighbors of~$v \in V$ wrt.~$E \subseteq \binom{V}{2}$, i.\,e., $N^+(v,E) := \{u \in V \mid (v,u)\in E\}$ \\
		$\Delta^-_{U}$, $\Delta^+_{U}$	& maximum in- and out-degree in~$U$ resp., e.\,g., $\Delta^-_U := \max_{u \in U} |N^-(u,E_A \cupdot E_P)|$ \\
		$\delta^-_{U}$, $\delta^+_{U}$	& minimum in- and out-degree in~$U$ resp., e.\,g., $\delta^+_U := \min_{u \in U} |N^+(u,E_A \cupdot E_P)|$ \\
		$\Delta_A^-$, $\delta_A^-$	& maximum resp. minimum number of papers per author \\
		$\Delta_P^+$, $\delta_P^+$ 	& maximum resp. minimum number of authors per paper \\
		$\Delta_A^+$, $\delta_A^+$	& maximum resp. minimum number of papers any author is qualified to review \\
		$\Delta_P^-$, $\delta_P^-$	& maximum resp. minimum number of potential reviewers for any paper \\
		\bottomrule
	\end{tabularx}
\end{table}

Using this notation, we define our central problem
and refer to \cref{tab:notation} for further necessary variable definitions: 
\decprob{[Weighted] \cfpr}
{A directed bipartite graph $(A\cupdot P,E_A\cupdot E_P)$ and non-negative integers $\ccr$, $\ddp$, and $z$ [and a weight function $w:E_A\mapsto \mathbb{Z}$ and an integer $W$].}
{Is there a $\ccr$-$\ddp$-valid and $z$-cycle-free review assignment $E'\subseteq E_A$ [of weight at least $W$, i.e., $\sum_{e\in E'} w(e)\geq W$]?}

\section{NP-Hardness in Various Restricted Cases}
\label{sec:hardness}
From the work of \citet[Theorem 4.12]{DBLP:conf/atal/BarrotLPS20} it follows that \cfpr is NP-hard in the single-author-single-paper setting ($\Delta_A^-=\Delta_P^+=1$) even if $\ccr=\ddp=1$ and $z=2$. 
However, as in reality instances of \cfpr are hardly arbitrary but have a quite strong structure, in this section we prove that the NP-hardness of \cfpr upholds even if the given instance fulfills further quite restrictive conditions, e.g., each agent is qualified to review all papers or our problem specific parameters ($\Delta_A^-,\Delta_P^+,\Delta_A^+,\Delta_P^-,\ccr,\ddp, z$) are small constants.  
\subsection{Sparse Review Graph and Small Weights}
We start by considering the case where all our parameters are small.
Specifically, we show the NP-hardness of \cfpr for arbitrarily $z\geq 2$ even if each paper is only authored by at most two agents, each agent authors at most two papers, each agent is only qualified to review at most three papers, and for each paper only at most three agents are qualified to review it (see \cref{tab:notation} for definitions).
\begin{restatable}{theorem}{cfprsparse}
	\label{th:cfpr-sparse}
	For any $z\geq 2$, \cfpr is NP-hard, even if $\Delta_A^+=\Delta_P^-=3$, $\Delta_A^-=\Delta_P^+= 2$, $n_A=n_P$, and $\ccr=\ddp=1$. 
	The hardness results still hold if agents are not allowed to review papers of co-authors.
\end{restatable}
\begin{proof}
	We reduce from an NP-hard variant of \textsc{Satisfiability} where each clause consists of exactly three literals and each variable occurs positive in at most two clauses and negative in at most two clauses \citep{DBLP:journals/eccc/ECCC-TR03-049}.  
	
	\paragraph{Construction.}
	Given an instance of \textsc{Satisfiability} consisting of a set $X$ of variables and a set $C$ of clauses, we set $\ddp=\ccr=1$ and $z$ to some integer greater than one. We construct the set $A$ of agents and the set $P$ of papers as follows.
	For each variable $x\in X$, we introduce three agents $a_x$, $a_{\bar{x}}$, and $b_x$ and three papers $p_x$, $p_{\bar{x}}$, and $q_x$ ($q_x$ has no author and can be considered as a dummy paper). 
	Agents $a_x$ and $b_x$ are qualified to review $p_x$, agents $a_{\bar{x}}$ and $b_x$ are qualified to review $p_{\bar{x}}$ and agents $a_x$ and $a_{\bar{x}}$ are qualified to review $q_x$. 
	Intuitively, either does $a_x$ review $p_x$ (which corresponds to setting $x$ to false) or $a_{\bar{x}}$ review $p_{\bar{x}}$ (which corresponds to setting $x$ to true).
	
	For each clause $c=\ell_1 \vee \ell_2 \vee \ell_3$, we introduce three agents $a_c^1$, $a_c^2$, and $a_c^3$ and three papers $p_c^1$, $p_c^2$, and $p_c^3$ where $a_c^i$ is qualified to review $p_c^i$ for $i\in [3]$. 
	Moreover, we introduce two dummy agents that are both qualified to review $p_c^1$, $p_c^2$, and $p_c^3$ and two dummy papers who $a_c^1$, $a_c^2$, and $a_c^3$ are all qualified to review. 
	Notably, for one $i\in [3]$, $a_c^i$ needs to review $p_c^i$ (which corresponds to $c$ being fulfilled because of~$\ell_i$). 
	
	Concerning the authors of each paper, for each clause $c=\ell_1 \vee \ell_2 \vee \ell_3$ and $i\in [3]$, $a_c^i$ is an author of $p_{\ell_i}$ and $a_{\ell_i}$ is an author of $p^i_{c}$. 
	
	It is easy to see that each agent is only qualified to review at most three papers and that for each paper only at most three agents are qualified to review it. 
	Moreover, as each literal only appears in at most two clauses, every paper has at most two authors and each agent authors at most two papers. 
	Moreover, note that $|A|=|P|$, implying that each agent has to review \emph{exactly} one paper.
	
	{\bfseries ($\Rightarrow$)} Let $Z$ be the set of variables that are set to true in a satisfying assignment of the given \textsc{Satisfiability} instance.
	Then, for $x\in Z$, we assign $b_x$ to $p_x$, $a_x$ to $q_x$, and $a_{\bar{x}}$ to $p_{\bar{x}}$, while for $x\notin Z$, we assign $b_x$ to $p_{\bar{x}}$, $a_{\bar{x}}$ to $q_x$, and $a_{x}$ to $p_{x}$. 
	For a clause $c=\ell_1 \vee \ell_2 \vee \ell_3$, let $\ell_{i^*}$ with $i^*\in [3]$ be a literal from $c$ that is set to true by the given assignment (such a literal exists because the given assignment is satisfying). 
	Then, we set $a_c^{i^*}$ to review $p_c^{i^*}$. The two dummy agents from this clause are assigned arbitrarily to $p_c^{i}$ for $i\in [3]\setminus \{i^*\}$ and the agents $a_c^i$ for $i\in [3]\setminus \{i^*\}$ are assigned arbitrarily to the two dummy papers. 
	To show that the constructed assignment does not contain a review cycle (of arbitrary length) note that only papers that have an author are papers $p_{\ell}$ for some literal $\ell$ (which are authored by $a_c^i$ for some $c\in C$ and $i\in [3]$ where $\ell$ appears in $c$ as the $i$th literal) and papers $p^i_{c}$ for some $c=\ell_1\vee \ell_2\vee \ell_3 \in C$ and $i\in [3]$ (which are authored by $a_{\ell_i}$). 
	Thus, every review cycle of length at least two needs to contain an agent $a_c^i$ for some $c\in C$ and $i\in [3]$ and $a_{\ell}$, where $\ell$ appears in $c$ as the $i$th literal, and $a_c^i$ reviews $p_c^i$ and $a_{\ell}$ reviews $p_{\ell}$. 
	For $a_c^i$ to review $p_c^i$ it needs to hold that the given assignment satisfies $\ell$. 
	However, by our construction of the review assignment, $a_{\ell}$ reviewing $p_{\ell}$ implies that $\bar{\ell}$ is satisfied. 
	Thus, no review cycle exists.
	
	{\bfseries ($\Leftarrow$)} Assume we are given a $1$-$1$-valid $z$-cycle-free review assignment.  
	Let $Y:=\{x\in X\mid a_{\bar{x}} \text{ reviews } p_{\bar{x}}\}$. 
	We claim that the assignment $\alpha$ which sets all variables in $Y$ to true and all variables in $X\setminus Y$ to false satisfies the given formula. 
	Assume for the sake of contradiction that there exists a clause $c=\ell_1\vee \ell_2 \vee \ell_3\in C$ which is not satisfied by $\alpha$. 
	As the given assignment is $1$-$1$-valid and we have the same number of agents and papers in the constructed instance, there is a $i^*\in [3]$ such that $a_c^{i^*}$ reviews $p_c^{i^*}$. Note that by the same reasoning, for each $x\in X$, either does $a_x$ review $p_x$ or $a_{\bar{x}}$ review $p_{\bar{x}}$. 
	Thus, if a literal $\ell$ is not satisfied by $\alpha$, then $a_{\ell}$ reviews $p_{\ell}$.
	As $\ell_{i^*}$ is not satisfied by $\alpha$ , $a_{\ell_{i^*}}$ reviews $p_{\ell_{i^*}}$. 
	Thus, $a_c^i$ and $a_{\ell_{i^*}}$ form a review cycle of length two, as $a_c^i$ reviews $p_c^i$, which is authored by $a_{\ell^*_i}$, and $a_{\ell^*_i}$ reviews $p_{\ell^*_i}$, which is authored by $a_c^i$, a contradiction. \end{proof}
	
	The above reduction crucially relies on the ``sparsity'' of the qualifications, i.e., that each agent is qualified to review between two and three papers and that for each paper only two or three agents are qualified to review it. 
	Motivated by the observation that, in practice, reviewers are typically qualified to review more than just two or three papers and that for each paper there typically exists more than just two or three qualified reviewers, it is a natural question whether our above hardness result still extends to this case. 
	We answer this question affirmative by proving hardness for arbitrary $\delta_A^+$ and $\delta_P^-$, i.e., for the case where each agent is qualified to review at least $\delta_A^+$ papers and for each paper there exist at least  $\delta_P^-$ agents that are qualified to review to:
	\begin{proposition} \label{pr:extension-Th1}
	 For any $z\geq 2$, $\delta_P^-\geq 2 \leq \delta_A^+$, \cfpr is NP-hard, even if  $\Delta_A^-=\Delta_P^+= 2$, $n_A=n_P$, and $\ccr=\ddp=1$. 
	\end{proposition}
	\begin{proof}
	Let $\delta:=\max(\delta_A^+,\delta_P^-)$.
	We reduce from the restricted NP-hard variant of \cfpr considered in \Cref{th:cfpr-sparse}. 
     Given an instance $\mathcal{I}=((A\cupdot P,E_A\cupdot E_P),\ccr=1,\ddp=1,z=2)$ of \cfpr with $\Delta_A^-=\Delta_P^+= 2$, we modify the instance $\mathcal{I}$ by  introducing two sets $A'$ and $A''$ of $\delta$ agents each and two sets $P'$ and $P''$ of $\delta$ papers each. All agents from $A'$ are qualified to review all papers from $P'$ and from $P$. In addition to being qualified to review some papers from $P$ (as captured in $E_A$), all agents from $A$ are qualified to review all papers from $P''$. Moreover, all agents from $A''$ are qualified to review all papers from $P''$. Thereby, all agents are qualified to review at least $\delta$ papers and for each paper at least $\delta$ agents are qualified to review it. Notably, we still have $|A|=|P|$. Thus, as agents from $A''$ are only qualified to review papers from $P''$ and $|A''|=|P''|$, all papers from $P''$ need to be reviewed by agents from $A''$ (which is always possible to do in without creating a review cycle as no paper from $A''$ has an author). Similarly, as papers from $P'$ can only be reviewed by agents from $A'$ and $|A'|=|P'|$, all agents from $A'$ need to review papers from $P'$ (which is always possible to do in without creating a review cycle as no paper from $A'$ has an author). Thus, all agents from $A$ need to review papers from $P$ from which the correctness of the reduction directly follows.
     
     Lastly, note that we did not modify the set of authors for any paper from $P$ and did not add papers with an author. Thus, it still holds in the modified instance that each agent authors at most two papers and each paper has at most two authors ($\Delta_A^-=\Delta_P^+= 2$). 
\end{proof}

While we prove hardness for arbitrary $\delta_A^+$ and $\delta_P^-$, in our construction from \Cref{pr:extension-Th1}, there are always agents that are not qualified to review ``many'' papers (around $\frac{2}{3}$) and always papers that cannot be reviewed by ``many'' agents (around $\frac{2}{3}$). 
Thus, interpreting a qualification as the absence of a conflict of interest, for our NP-hardness agents need to have many conflicts. 
In \Cref{sec:heuristic}, we prove that this does not happen by accident, as if the number of conflicts per agent/paper (and $\Delta_A^-$, $\Delta_P^+$, $\ccr$, and $\ddp$) are ``small'', then \cfpr always admits a solution.

In \wcfpr it is possible to encode the ``qualifications'' of agents into weights:
If we modify the reduction from above and give an agent-reviewer pair weight one if the agent is qualified to review the paper and weight zero otherwise, we get that \wcfpr is NP-hard even if each agent is qualified to review all papers and we have few non-zero weights.
\begin{corollary}\label{co:weights}
 For any $z\geq 2$, \textsc{Weighted Cycle-Free Peer Reviewing} is NP-hard, even if each agent is qualified to review all papers, each agent gives only at most three papers a non-zero weight, for each paper at most three agents give it a non-zero weight,  $\Delta_P^+\leq 2\geq \Delta_A^-$, $n_A=n_P$, and $\ccr=\ddp=1$. 
\end{corollary}

\subsection{No Conflicts of Interest}
We now extend the hardness from \Cref{co:weights} for the case where each agent is qualified to review all papers (no conflicts) to the unweighted case. However, our new reduction relies on the existence of papers with many authors and agents authoring many papers.

To show that \cfpr is NP-hard even if each agent is qualified to review all papers, $n_A=n_P$, $\ccr=\ddp=1$, and $z=2$ (\Cref{th:full}), we reduce from \textsc{Multicolored Independent Set} where we are given a graph $G$ with vertices partitioned into $k$ sets $V^1,\dots, V^k$ (to which we refer as color classes) and the question is whether there exists a subset of $k$ vertices, containing one vertex from each class, that are pairwise non-adjacent. 
We denote as $n:=|V^1|$ the number of vertices in the first color class and assume without loss of generality that $n>k$ and that $|V^c|:=n+c-1$ for $c\in [k]$ (note that we can do so because we can always add vertices that are connected to all other vertices and put them into one of the color classes).

\paragraph{Construction.} Given an instance $\mathcal{I}$ of \textsc{Multicolored Independent Set} $G=(V=(V^1,\dots, V^k),E)$, we construct an instance $\mathcal{I}'$ of \cfpr as follows. 
For each color $c\in [k]$, we add a \emph{special agent} $a^c_*$ and a \emph{special paper} $p^c_*$. 
Moreover, for each vertex $v\in V^c$, we add a \emph{vertex agent} $a^c_v$ and a \emph{vertex paper} $p^c_v$. 
Further, we add $n+c-2$ \emph{dummy agents} $\tilde{a}^c_1, \dots, \tilde{a}^c_{n+c-2}$ and $n+c-2$ \emph{dummy papers} $\tilde{p}^c_1, \dots, \tilde{p}^c_{n+c-2}$. 
Lastly, we insert an agent $a^*$ and a paper $p^*$. 

The paper $p^*$ is authored by all vertex agents and dummy agents. 
For color $c\in [k]$, $p^c_*$ is authored by all vertex und dummy agents from colors $c'\neq c\in [k]$ and agent $a^*$. 
Further, all dummy papers $\tilde{p}^c_i$ for $i\in [n+c-2]$ are authored by the special agent $a^c_*$. 
For a vertex $v\in V^c$, paper $p^c_v$ is authored by the special agent $a^c_*$, all agents corresponding to vertices from $V^c\setminus \{v\}$ or to vertices adjacent to $v$ in $G$, i.e., $p^c_v$ is authored by agents $\{a^c_*\}\cup\{a^c_{v'}\mid v\neq v'\in V^c\}\cup \{a^{c'}_{v'} \mid c'\in [k], v'\in V^{c'}, \{v,v'\}\in E \}$. 
Each agent is qualified to review all papers and we set $\ccr=\ddp=1$ and $z=2$.

\begin{lemma} \label{le:fd}
	If the given instance $\mathcal{I}$ of \textsc{Multicolored Independent Set} is a YES-instance, then the constructed instance $\mathcal{I}'$ of \cfpr is a YES-instance.
\end{lemma}
\begin{proof}
	Let $V'=\{w^1, \dots w^k\}\subseteq V$ be a independent set of size $k$ in the given \textsc{Multicolored Independent Set} instance $\mathcal{I}$ with $w^c\in V^c$ for $c\in [k]$. 
	From this we construct a solution for the constructed \cfpr instance $\mathcal{I}'$ as follows. Agent $a^*$ reviews paper $p^*$. 
	For $c\in [k]$, special agent $a^c_*$ reviews special paper $p^c_*$. 
	Vertex agents $\{a^c_{v'} \mid v'\in V^c\setminus \{w^c\}\}$ are assigned arbitrarily to dummy papers $\tilde{p}^c_1, \dots, \tilde{p}^c_{n+c-2}$. 
	Lastly, vertex agent $a^c_{w^c}$ reviews paper $p^c_{w^c}$ and the dummy agents from class $c$ are assigned arbitrarily to the remaining vertex papers from this class. 
	Note that by construction, the described assignment is $1$-$1$ valid. 
	Moreover, it is easy to verify that no agent reviews a paper authored by it so it remains to check for reviewing cycles of length two. 
	All special agents are only authors of papers from their color class but review papers authored solely by agents outside their color class. 
	Thus there exist no review cycles involving special agents. 
	All papers $a^*$ wrote are reviewed by special agents so $a^*$ cannot be part of a review cycle. 
	Dummy agents only write papers that are reviewed by special agents and $a^*$ so no dummy agent can be part of a review cycle. 
	Thus, every possible review cycle of length two needs to involve two vertex agents. 
	As no dummy paper is written by a vertex agent, the only vertex agents that review papers authored by other vertex agents are those assigned to vertex papers, i.e., agents $\{a^1_{w^1}, \dots a^k_{w^k}\}$. 
	Assume for the sake of contradiction that $a^i_{w_i}$ (which reviews paper $p^i_{w_i}$) forms a cycle with reviewer $a^{i'}_{w^{i'}}$ with $i\neq i' \in [k]$. 
	However, from this it follows by the definition of a review cycle that $a^{i'}_{w^{i'}}$ is an author of paper $p^i_{w_i}$, which implies that $\{w^i, w^{i'}\}\in E$ contradicting that $V'$ is an independent set. 
\end{proof}

We now turn to proving the backwards direction of the reduction. To do this, we first identify several assignments that need to be made in all solutions to the constructed \cfpr instance. We start by proving that $a^*$ needs to review $p^*$.  
\begin{lemma} \label{le:a*}
	In every $1$-$1$ valid $2$-cycle-free assignment in the constructed instance $\mathcal{I}'$, $a^*$ reviews $p^*$. 
\end{lemma}
\begin{proof}
	Recall that all agents except all special agents and agent $a^*$ are authors of $p^*$.  
	So for the sake of contradiction let us assume that special agent $a_*^c$ for some $c\in [k]$ reviews $p^*$. 
	However, to prevent a reviewing cycle, this implies that only the remaining $k-1$ special agents and $a^*$  can review papers written by $a_*^c$. 
	However, as $a_*^c$ is an author of all vertex papers corresponding to vertices from $V^c$ and we have assumed that each set $V^c$ consists of more than $k$ vertices, these $k$ agents are not enough to review all papers written by $a_*^c$, a contradiction.
\end{proof}

We next prove that $a^c_*$ reviews $p^c_*$ for all $c\in [k]$. For this, we need the following lemma:
\begin{lemma} \label{le:ac*}
	In every $1$-$1$ valid $2$-cycle-free assignment in the constructed instance $\mathcal{I}'$, if $a^c_*$ reviews paper $p^{c'}_*$ for $c,c'\in [k]$, then only vertex and dummy agents from class $c'$ and special agents can review dummy and vertex papers from class $c$.
\end{lemma}
\begin{proof}
	Note that the special agent $a^c_*$ is an author of all dummy and vertex papers from color class $c$. 
	Moreover, paper $p^{c'}_*$ is authored by all dummy and vertex agents from color classes different from $c'$. 
	Thus, if $a^c_*$ reviews $p^{c'}_*$, then no vertex or dummy agent from a class different from $c'$ can review papers written by $a^c_*$. 
	As $a^c_*$ authors all dummy and vertex papers from class $c$, the lemma follows.
\end{proof}

Using this, we are able to prove that each special agent reviews the corresponding special paper.

\begin{lemma} \label{le:ac*-2}
	In every $1$-$1$ valid $2$-cycle-free assignment in the constructed instance $\mathcal{I}'$, for $c\in [k]$, $a^c_*$ reviews $p^c_*$. 
\end{lemma}
\begin{proof}
	By \Cref{le:a*}, $a^*$ is assigned to $p^*$, which is authored by all dummy agents and vertex agents. 
	Thus, to prevent the existence of reviewing cycles, only special agents can review papers written by $a^*$. 
	As for each $c\in [k]$, $p^c_*$ is written by $a^*$, it follows that the set of $k$ agents $\{a_*^c\mid c\in [k]\}$ needs to review the set of $k$ papers $\{p_*^c\mid c\in [k]\}$. 
	For the sake of contradiction, let us assume that special agent $a^c_*$ reviews paper $p^{c'}_*$ for $c\neq c'\in [k]$. 
	We assume without loss of generality that $c'<c$ (if there exists a pair where $a^{\tilde{c}}_*$ reviews paper $p^{\tilde{c}'}_*$ with $\tilde{c}<\tilde{c}'$ there also has to exist one with $c'<c$).
	By \Cref{le:ac*} and as special agents need to review special papers, from this it follows that only dummy and vertex agents from color $c'$ can review the vertex and dummy agents from class $c$ (which are all written by $a^c_*$). 
	As we have assumed that $c'<c$, the number of these agents ($2n+2c'-3$) does not suffices to review all of these papers ($2n+2c-3$), a contradiction. 
\end{proof}

We are now ready to prove the correctness of the backwards direction of the reduction:
\begin{lemma}\label{le:bd}
	If the constructed instance $\mathcal{I}'$ of \cfpr  is a YES-instance, then the given instance $\mathcal{I}$ of \textsc{Multicolored Independent Set} is a YES-instance.
\end{lemma}
\begin{proof}
	From \Cref{le:a*}, \Cref{le:ac*}, and \Cref{le:ac*-2} it follows that for each color $c\in [k]$ every vertex and dummy agent from this color class needs to review a vertex or dummy paper from this color class and that each vertex or dummy paper from this color class needs to be reviewed by a vertex or dummy agent from this color class. 
	As there exist $n+c-2$ dummy agents from color class $c$ but $n+c-1$ vertex papers at least one  vertex paper from color class $c$ needs to be reviewed by a vertex agent from color class $c$. 
	Note that for each $v\in V^c$, agent $a_v^c$ is an author of all vertex papers except $p_v^c$. 
	Thus, for each color $c\in [k]$ there needs to exist (at least) one agent $a_{w_c}^c$ for some $w_c\in V^c$ that reviews $p_{w_c}^c$. 
	So let $a_{w_1}^1,\cdots, a_{w_k}^k$ be a list of those agents (containing one vertex agent from each color class). 
	We claim that $\{w_1,\dots , w_k\}$ forms an independent set in $G$. 
	For the sake of contradiction assume that $\{w_c,w_{c'}\}\in E$ for $c\neq c'\in [k]$, then by construction it follows that $a_{w_{c}}^c$ who reviews paper $p_{w_{c}}^c$ is an author of paper $p_{w_{c'}}^{c'}$ and similarly $a_{w_{c'}}^{c'}$ who reviews $p_{w_{c'}}^{c'}$ is an author of $p_{w_{c}}^c$. 
	Thus, $a_{w_{c}}^c$ and $a_{w_{c'}}^{c'}$ form a reviewing cycle, a contradiction.
\end{proof}

From \Cref{le:fd} and \Cref{le:bd}, \Cref{th:full} directly follows:
\begin{restatable}{theorem}{full}
	\label{th:full}
	\cfpr is NP-hard even if each agent is qualified to review all papers, $n_A=n_P$, $\ccr=\ddp=1$, and $z=2$.
\end{restatable}

The reduction from \Cref{th:full} heavily relies on the possibility that an agent reviews a paper written by an agent with whom she has a joint paper. 
As some conferences might declare an automatic conflict of interest for co-authors, we now consider the case where an agent is qualified to review all papers that are not authored by one of her co-authors:

\begin{restatable}{theorem}{cof}
	\label{th:cof}
	 \cfpr is NP-hard even if each agent is qualified to review all papers that are not written by one of her co-authors, $\ccr=\ddp=1$, and $z=2$. 
\end{restatable}
\begin{proof}
	We reduce from \cfpr with $\ccr=\ddp=1$, and $z=2$  where agents are not qualified to review papers of co-authors, which is NP-hard as proven in \Cref{th:cfpr-sparse}. We assume without loss of generality that for each paper there is one agent who is not qualified to review it.
	
	\paragraph{Construction.} Given an instance $\mathcal{I}=((A\cupdot P,E_A\cupdot E_P),\ccr=1,\ddp=1,z=2)$ of \cfpr, we construct a new instance $\mathcal{I}'$ with agents $A'$ and papers $P'$ and $\ccr'=\ddp'=1$ and $z'=2$. We start by setting $A'=A$. Next, we add agents $w$, $x$, $y$, and $z$ to $A'$. For each agent $a\in A$ and each paper $p\in P$, we insert an agent $a_p$ to $A'$ and add a so called \emph{agent} paper which is authored by $a$ and $a_p$ to $P'$. For each paper $p\in P$, we introduce an agent $b_p$ to $A'$. Moreover, we introduce $n_A\cdot n_P$ dummy agents $d_1,\dots, d_{n_A\cdot n_P}$ to $A'$. 
	We introduce five different (types of) papers in $P'$: 
	\begin{itemize}
		\item For each paper $p\in P$, we introduce a paper $p$ to $P'$ that is written by all authors of $p$, agent $b_p$ and by agents $a_p$ for all agents $a\in A$ that are not qualified to review $p$ in $\mathcal{I}$.  
		\item We introduce a paper $q$ authored by $w$ and $a_p$ for all $a\in A$ and $p\in P$.
		\item We introduce a paper $q'$ authored by $x$ and $a_p$ for all $a\in A$ and $p\in P$ to $P'$. 
		\item We introduce a paper $r$ authored by $y$ and all dummy agents $d_1,\cdots, d_{n_A\cdot n_P}$ and $b_p$ for each $p\in P$.
		\item  We introduce a paper $r'$ authored by $z$ and all dummy agents $d_1,\cdots, d_{n_A\cdot n_P}$ and $b_p$ for each $p\in P$.
	\end{itemize}
	Each agent is qualified to review all papers that are not written by one of her co-authors.
	
	{\bfseries ($\Rightarrow$)} 
	Given a $1$-$1$-valid $2$-cycle-free review assignment for $\mathcal{I}$, we construct a $1$-$1$-valid $2$-cycle-free review assignment for $\mathcal{I}'$ as follows. 
	All agents from $A$ still review the same papers as in the given assignment (which are all still qualified to do so because we have not added or removed any papers with two authors from $A$ apart from copies of papers from $P$). 
	Agent $w$ reviews $r$, agent $x$ reviews $r'$, $y$ reviews $q'$, and $z$ reviews $q$ (which are all are qualified to do so). Moreover, the dummy agents are assigned arbitrarily to the agent papers, which they are qualified to review because dummy agents only author papers together with agents $\{b_p\mid p\in P\}$ and $y$ and $z$. 
	
	Concerning review cycles, note that agents $\{a_p, b_p\mid a\in A,$ $ p\in P\}$ do not review any paper. Moreover dummy agents only review papers written by agents  $\{a_p, a \mid a\in A,$  $p\in P\}$ but are only reviewed by $x$ and $w$ and thus cannot be part of a review cycle. Moreover, also no agent from $A$ can be part of a review cycle because there was no such cycle in the given review assignment and all agents that author a paper reviewed by an agent from $A$ are not part of a review cycle. Thus, any review cycle needs to consists of $w$, $x$, $y$, and $z$. Note that $w$ reviews a paper of $y$, $y$ reviews a paper of $x$, $x$ reviews a paper of $z$, and $z$ reviews a paper of $w$. Thus, these agents form a $4$-cycle but no $2$-cycle.  
	
	{\bfseries ($\Leftarrow$)} 
	Given a $1$-$1$-valid $2$-cycle-free review assignment for $\mathcal{I'}$, we claim that this assignment restricted to agents from $A$ and papers from $P$ is a solution to the given instance $\mathcal{I}$. To prove this, we will argue for all agents from $A'\setminus A$ that they cannot review a paper from $P$ from which the correctness directly follows, as we have not added any authors from $A$ to papers from $P$. Fix some paper $p'\in P$. We now iterate over all agents $A'\setminus A$ and argue why they cannot review $p'$. As we have assumed in $\mathcal{I}$ that for all papers there is an agent not qualified to review it, it follows that $p'$ has an author ${a^*}_{p'}$ for some $a^*\in A$ and author $b_{p'}$ in $\mathcal{I'}$. For agent $w$ and $x$ it holds that both have a joint paper with ${a^*}_{p'}$ and thus cannot review $p'$. Next, note that as, for each $a\in A$ and $p\in P$, $a_p$ is either identical to ${a^*}_{p'}$ or has a joint paper with ${a^*}_{p'}$ none of these agents can review $p'$. Lastly, $d_i$ for some $i\in [n_A\cdot n_P]$, $b_p$ for $p\in P \setminus \{p'\}$, and $y$ and $z$ have a joint paper with $b_{p'}$, which is an author of $p'$. Thus, all these agents (and thereby no agent from $A'\setminus A$) can review $p'$. 
	
\end{proof}

\subsection{Single-Author-Single-Paper Setting}
In their theoretical analysis, \citet{DBLP:conf/atal/BarrotLPS20} focus on \cfpr where each agent writes a single-author paper (we speak of an agent and its paper interchangeably) and qualifications are symmetric, i.e., if an agent $a$ is qualified to review agent $b$, then $b$ is qualified to review $a$. 
They prove that this problem is NP-hard for $\ccr=\ddp=1$ and $z=5$ (without bounds on $\Delta_A^+$ or $\Delta_P^-$) but polynomial-time solvable for arbitrary $\ccr=\ddp$ for $z=2$. 
We close the gap between these two results and extend their general picture by proving that for $\ccr=\ddp=2$, \cfpr is NP-hard for $z=3$ even if qualifications are symmetric and each agent is only qualified to review four agents, i.e., we need to decide for each agent $a$ which two of these four agents review $a$ and which two of these agents will get a review from~$a$.
\begin{restatable}{theorem}{otocfprNP}
	\label{th:otocfprNP}
	\cfpr is NP-hard, even if $z=3$,  $\ccr=\ddp=2$, $\Delta_A^-=\Delta_P^+=1$, $n_A=n_P$, each agent is qualified to review exactly four papers and if an agent $a$ can review the paper written by agent $b$, then $b$ can review the paper of $a$.
\end{restatable}

To prove \Cref{th:otocfprNP}, we reduce from \textsc{Two-in-Four-Satisfiability}, a variant of \textsc{Satisfiability}, where given a propositional formula $\varphi$ over variables $X$ where each clause contains four different literals, the question is whether there exists an assignment $\alpha$ of variables $X$ such that in each clause exactly two out of four literals are satisfied. As to the best of our knowledge, this variant of \textsc{Satisfiability} has not been considered before, we start by proving that it is NP-hard even if each literal appears exactly twice positive and twice negative: 
\begin{proposition}
	\textsc{Two-in-Four-Satisfiability} is NP-hard, even if each variable appears exactly twice positive and twice negative. 
\end{proposition}
\begin{proof}
	In \textsc{Monotone Not-All-Equal 3-Sat}, we are given a propositional formula where each clause contains three different positive literals and the question is whether there is a variable assignment such that in each clause at least one literal is set to true and at least one is set to false. Reducing \textsc{Monotone Not-All-Equal 3-Sat} to \textsc{Two-in-Four-Satisfiability} (without any additional restrictions) is straightforward: 
	Given an instance of \textsc{Monotone Not-All-Equal 3-Sat}, for each clause, we introduce a new variable which we add to the clause. 
	Thereby, we can extend a valid assignment $\alpha$ for the \textsc{Monotone Not-All-Equal 3-Sat} instance by setting for a clause the newly introduced variable to true if $\alpha$ originally sets only one literal from this clause to true and to false if $\alpha$ originally sets only one literal from this clause to false. The reverse direction is immediate.
	However, to achieve that each variable appears twice positive and twice negative, a slightly more involved approach is needed. 
	
	In fact, for simplicity, we reduce from the NP-hard variant of \textsc{Monotone Not-All-Equal 3-Sat} where each variable appears in exactly four clauses \citep{DBLP:journals/tcs/DarmannD20}. 
	Given an instance $\varphi=C_1 \wedge \dots \wedge C_m$ over variables $X$  of \textsc{Monotone Not-All-Equal 3-Sat}, note that $m$ needs to be even, as there are $m=\frac{4\cdot |X|}{3} $ and $m$ needs to be an integer.
	We now construct a new propositional formula $\phi$ over variable set $X'$ as follows. 
	For each clause $C_i=w\vee x \vee y$ for $i\in [m]$, we add variables $w_i$, $x_i$, $y_i$, and $z_i$ to $X'$ and clauses $w_i \vee x_i \vee y_i \vee z_i$ and $\overline{w_i} \vee \overline{x_i} \vee \overline{y_i} \vee \overline{z_i}$ to $\phi$. 
	Now, every variable appears once negative and once positive. 
	It remains to link the copies of each variable. 
	
	We do this for each variable separately. Let $x\in X$ be some original variable 
	and let $j_1, \dots, j_{4}$ denote the list of all clauses where $x$ 
	appears 
	in $\varphi$. We introduce dummy variables $a^1_x$ and $a^2_x$ to $X'$ and add clauses 
	$\overline{x_{j_1}}\vee
	x_{j_{2}} \vee a^1_x \vee \overline{a^1_x}$, $\overline{x_{j_2}}\vee 
	x_{j_{3}} \vee a^1_x \vee \overline{a^1_x}$, $\overline{x_{j_3}}\vee 
	x_{j_{4}} \vee a^2_x \vee \overline{a^2_x}$, and $\overline{x_{j_4}}\vee 
	x_{j_{1}} \vee a^2_x \vee \overline{a^2_x}$ to $\phi$. 
	As for each $j\in [3]$, exactly one of $\overline{x_{j_i}}$ and $x_{j_{i+1}}$ need to be set to true, these clauses enforce that $x_{j_1}$, $x_{j_2}$, $x_{j_3}$, and $x_{j_4}$, all have the same truth value. 
	Lastly, for $i\in [\frac{m}{2}]$, we add twice the clause $z_{2i-1} \vee \overline{z_{2i-1}} \vee z_{2i} \vee \overline{z_{2i}} $ which are always trivially satisfied.
	
	The correctness of the reduction is immediate and all variables appear twice positive and twice negative in $\phi$.
\end{proof}

Using this, we are now ready to prove \Cref{th:otocfprNP}:

\begin{proof}[Proof of \cref{th:otocfprNP}]
	We reduce from 	\textsc{Two-in-Four-Satisfiability} where each variable appears exactly twice positive and twice negative.
	
	\paragraph{Construction.} Given an instance of \textsc{Two-in-Four-Satisfiability} consisting of a propositional formula $\varphi=C_1 \wedge \dots \wedge C_m$ over variables $X=\{x_1,\dots, x_n\}$, for $i\in [n]$, we denote as $t^\text{pos}_{i,1}$ and $t^\text{pos}_{i,2}$ the indices of the two clauses in which variable $x_i$ appears positive and as $t^\text{neg}_{i,1}$ and $t^\text{neg}_{i,2}$ the indices of the two clauses in which variable $x_i$ appears negative. 
	From this, we construct an instance of \otocfpr as follows. 
	For $i\in [n]$, we introduce four agents $a^\text{pos}_i$, $a^\text{neg}_i$, $a^1_i$, and $a^2_i$ (constituting a gadget modeling this variable). 
	Moreover, for $j\in [m]$, we introduce one agent $b_j$.  
	Qualification are symmetric, i.e., if agent $a$ is qualified to review $b$, then $b$ is qualified to review $a$.
	For $i\in [n]$, $a^\text{pos}_i$ is qualified to review $a^1_i$, $a^2_i$, $b_{t^\text{pos}_{i,1}}$, and $b_{t^\text{pos}_{i,2}}$ (and the other way round). 
	Moreover, $a^\text{neg}_i$ is qualified to review $a^1_i$, $a^2_i$, $b_{t^\text{neg}_{i,1}}$, and $b_{t^\text{neg}_{i,2}}$ (and the other way round). 
	Lastly, $a^1_i$ and $a^2_i$ are qualified to review each other and $a^2_i$ and $a^1_{i+1}$ are qualified to review each other (where $i$ is taken modulo $n$).
	We set $\ccr=\ddp=2$ and $z=3$.  
	
	{\bfseries ($\Rightarrow$)} 
	Let $\alpha$ be an assignment of variables in $X$ that is a solution to the given \textsc{Two-in-Four-Satisfiability} instance. 
	For $i\in [n-1]$, we let $a^1_i$ review $a^2_i$ and $a^2_i$ review $a^1_{i+1}$. 
	Moreover, we let $a^1_{n}$ review $a^2_{n}$ and $a^2_n$ review $a^1_1$. 
	
	For $i\in [n]$ where $x_i$ is set to true by $\alpha$, we let $a^1_{i}$ and $a^2_i$ review $a^\text{pos}_i$ and $a^\text{neg}_i$ review $a^1_{i}$ and $a^2_i$.
	Moreover we let $a^\text{pos}_i$ review $b_{t^\text{pos}_{i,1}}$ and $b_{t^\text{pos}_{i,2}}$ and we let $b_{t^\text{neg}_{i,1}}$ and $b_{t^\text{pos}_{i,2}}$ review $a^\text{neg}_i$. 
	Conversely, for $i\in [n]$ where $x_i$ is set to false by $\alpha$, we let $a^1_{i}$ and $a^2_i$ review $a^\text{neg}_i$, we let $a^\text{pos}_i$ review $a^1_{i}$ and $a^2_i$. 
	Moreover, we let $a^\text{neg}_i$ review $b_{t^\text{neg}_{i,1}}$ and $b_{t^\text{neg}_{i,2}}$ and let $b_{t^\text{pos}_{i,1}}$ and $b_{t^\text{pos}_{i,2}}$ review~$a^\text{pos}_i$.
	
	As $\alpha$ sets exactly two literals in each clause to true and two to false, for each $j\in [m]$, $b_j$ is reviewed by two agents and reviews two agents. 
	The same also holds for all other agents, implying that the constructed review assignment is $2$-$2$ valid.  It is easy to see that there are no $2$-cycles. 
	Moreover, as no two agents $b_i$ and $b_j$ for $i\neq j$ are qualified to review each other and, for no $j\in [m]$, are there two agents that are both qualified to review $b_j$ and that are qualified to review each other, each $3$-cycle needs to solely consist of agents from a gadget corresponding to a single variable.
	So let us fix some $i\in [n]$. 
	The only possible $3$-cycles consist of $a^\text{pos}_i$, $a^1_i$ and $a^2_i$ or $a^\text{neg}_i$, $a^1_i$ and $a^2_i$. 
	However, there is no such $3$-cycle, as either $a^\text{pos}_i$ reviews both $a^1_i$ and $a^2_i$ and both $a^1_i$ and $a^2_i$ review $a^\text{neg}_i$, or $a^\text{neg}_i$ reviews both $a^1_i$ and $a^2_i$ and both $a^1_i$ and $a^2_i$ review $a^\text{pos}_i$. 
	Thus, the constructed assignment is $3$-cycle-free.
	
	{\bfseries ($\Leftarrow$)} Assume we are given a $2$-$2$-valid $3$-cycle-free review assignment in the constructed \cfpr instance. 
	Assume that $a^2_n$ reviews $a^1_1$ in the given assignment (if $a^1_1$ reviews $a^2_n$ an analgous argument works). 
	We now argue that $a^1_1$ needs to review $a^2_1$. 
	Assume for the sake of contradiction that this is not the case, then as $a^1_1$ is reviewed by $a^2_n$ and $a^2_1$, she needs to review $a^\text{pos}_1$ and $a^\text{neg}_1$. 
	However, to prevent a $3$-cycle, $a^2_1$ then needs to review $a^\text{pos}_1$ and $a^\text{neg}_1$, a contradiction (as $a^2_1$ gives three reviews). 
	Next, we want to argue that $a^2_1$ reviews $a^1_2$. 
	For the sake of contradiction, assume that $a^1_2$ reviews $a^2_1$. 
	Then, $a^2_1$ already gets two reviews and thus needs to review $a^\text{neg}_1$ and $a^\text{pos}_1$. 
	However, as $a_1^1$ already reviews $a_1^2$ either $a^\text{neg}_1$ or $a^\text{pos}_n$ review $a_1^1$ which leads to a $3$-cycle together with $a_1^1$ and $a_1^2$. 	
	Applying the same arguments inductively, it follows that for $i\in [n-1]$, $a^1_i$ review $a^2_i$ and $a^2_i$ reviews $a^1_{i+1}$ and that $a^1_{n}$ reviews $a^2_{n}$. 

	Further, observe that for each $i\in [n]$ agents $a_i^1$ and $a_i^2$ either both review $a_i^\text{pos}$ or both get reviews from $a_i^\text{pos}$.
	For the sake of contradiction, assume that this is not the case. 
	If $a^2_i$ reviews $a_i^\text{pos}$ and $a_i^\text{pos}$ reviews $a^1_i$, then we have a $3$-cycle consisting of these three agents. 
	Otherwise, $a^1_i$ reviews $a_i^\text{pos}$ and $a_i^\text{pos}$ reviews $a^2_i$. 
	However, as the given assignment is $2$-$2$-valid, from this it follows that $a^2_i$ reviews $a^\text{neg}_i$ and $a^\text{neg}_i$ reviews $a^1_i$, which leads to a $3$-cycle consisting of $a^1_i$, $a^2_i$, and $a^\text{neg}_i$. 
	Thus, we have reached a contradiction proving our initial claim. 
	Moreover, as the given assignment is $2$-$2$ valid, in case that $a_i^1$ and $a_i^2$ both review $a_i^\text{pos}$, then $a_i^\text{neg}$ reviews $a_i^1$ and $a_i^2$, and in case that $a_i^\text{pos}$ reviews both $a_i^1$ and $a_i^2$, then $a_i^1$ and $a_i^2$ both review $a_i^\text{neg}$. 
	We now construct an assignment $\alpha$ by, for $i\in [n]$, setting variable $x_i$ to true if $a_i^1$ and $a_i^2$ review $a_i^\text{pos}$ and $x_i$ to false if $a_i^1$ and $a_i^2$ review $a_i^\text{neg}$. 
	Using our argument from above, it follows that $\alpha$ is well-defined. 
	Moreover, the given assignment is $2$-$2$-valid, if $\alpha$ sets a literal to true, then the agents corresponding to this literal review the agents corresponding to the two clauses in which the literal appears. 
	Similarly, if $\alpha$ sets a literal to false, then the agents corresponding to this literal get a review from the two agents corresponding to the two clauses in which the literal appears. 
	Thus, as each agent corresponding to a clause gets and issues two reviews (as the given assignment is $2$-$2$-valid), it follows that $\alpha$ sets for each clause exactly two literals to true and thus that $\alpha$ is a solution to the given instance of \textsc{Two-in-Four-Satisfiability}.
\end{proof}

\section{Polynomial-Time Solvable Special Cases}\label{sec:heuristic}

In this section, we identify conditions under which a short-cycle-free review assignment provably exists and can be computed in polynomial time.
As we will see in our experiments, the subsequently presented algorithms provide short-cycle-free review assignments even beyond the theoretical limitations we discuss below.
As we are interested in computing $z$-cycle-free review assignments for $z\geq 1$, no author is allowed to review one of its own papers.
That is why throughout this section we assume that we do not have~$(a,p) \in E_A$ and~$(p,a)\in E_P$ at the same time.

Our algorithms in this section are based on the following simple observation:
Given a partial $z$-cycle-free review assignment~$E'$ and a paper~$p \in P$ that requires more assigned reviewers, the number of potential reviewers that would create a $z$-cycle--if assigned to review~$p$--is bounded by a function in~$z$, the maximum number~$\Delta_P^+$ of authors per paper, and the maximum number~$\ccr$ of reviews per agent; the precise function is given in the subsequent proofs.
Thus, assuming that the minimum number~$\delta_P^-$ of potential reviewers for each paper is large compared to~$z$, $\Delta_P^+$, $\ddp$, and~$\ccr$, for each paper~$p$ there are always reviewers that can be assigned to review~$p$ without creating a $z$-cycle.
Note that in practice we can expect that~$z$, $\ddp$, and~$\ccr$ are quite small.
Moreover,  while the minimum number of fitting reviewers might be not very large, it is not uncommon to assign papers to reviewers that are not ``perfect''.
Thus, interpreting~$\delta_P^-$ as the number of community members that do not have a conflict of interest actually yields relative large values for~$\delta_P^-$ in practice.

We start with a very restrictive setting and then, step by step, generalize the approach and the results.
First, each paper is written by exactly one author, each agent has at most one paper and we want a completely cycle-free review assignment (i.\,e., $z$-cycle-free for every~$z\in \N$).
This of course implies that some agents cannot be authors of papers and so the number~$n_P$ of papers is smaller than the number~$n_A$ of agents.
However, it allows \cref{alg:greedy-dag} to work (implicitly) with the topological ordering of the (acyclic) review assignment while constructing it. 

\begin{algorithm}[t!]
	\caption{A greedy algorithm computing a $\ddp$-$\ddp$-valid completely cycle-free assignment~$E'$.
	}
	\label{alg:greedy-dag}
	$E' \gets \emptyset; S_0 \gets$agents without papers \label{line:initialize}\;
	\tcc{$\phi_i(a)$ is the free reviewing capacity of $a$ before iteration~$i$ of the for-loop from Line 3; each agent reviews at most~$\ddp$ papers}
	\lForEach{$a \in A$}{$\phi_0(a) \gets \ddp$}
	
	\tcc{Assign reviewers to one paper per iteration:}
	\For(\label{line:paper-loop}){$i \gets 0$ \KwTo $n_P - 1$}
	{
		\lForEach{$a \in A$}{$\phi_{i+1}(a) \gets \phi_{i}(a)$}
		select some $(p,a) \in E_P$ where~$p$ has no reviews yet \label{line:select-paper} \;
		\tcc{collect qualified reviewers and assign~$\ddp$ of them to~$p$}
		$R \gets \{b \in S_i \mid (b,p) \in E_A\}$ \label{line:eligible-reviewers}\;
		\For(\label{line:assign-reviewers}){$j \gets 1$ \KwTo $\ddp$} 
		{
			arbitrary~$b \in R$ reviews~$p\colon E' \gets E' \cup \{(b,p)\}$\;
			$\phi_{i+1}(b) \gets \phi_i(b) - 1$; $R \gets R \setminus \{b\}$ \label{line:assign-reviewers-end}
		}
		\tcc{collect possible reviewers for next paper}
		$S_{i+1} := \{a\} \cup  \{b \in S_i \mid \phi_{i+1}(b) > 0\}$ \label{line:update-S-set}
	}
	\textbf{return} $E'$
\end{algorithm}%

\begin{proposition}\label{lem:cycle-free-simple-setting}
	If~$\Delta_A^- \le 1 = \delta_P^+ = \Delta_P^+$, $\ddp \le \ccr$, and~$\delta_P^- \ge n_P + \ddp$, then \cref{alg:greedy-dag} computes a $\ddp$-$\ddp$-valid and completely cycle-free review assignment in linear time.
\end{proposition}

\begin{proof}
	We first show the correctness of \cref{alg:greedy-dag}.
	Clearly, if in each iteration of the loop in \cref{line:paper-loop} the set of eligible reviewers~$R$ (see \cref{line:eligible-reviewers}) is of size at least~$\ddp$, then a completely cycle-free review assignment is created as each agent only reviews papers from agents ``occurring'' later during the algorithm.
	Observe that if~$|S_i| \ge n_A - \delta_P^- + \ddp$ for~$i \in \{0,\ldots,n_P-1\}$, then in iteration~$i$ we have~$|R| \ge \ddp$: 
	There are at most~$n_A - \delta_P^-$ agents in~$S_i$ that cannot review~$p$ (the corresponding edge is not in~$E_A$) and, thus, at least~$\ddp$ agents in~$S_i$ are eligible to review~$p$.
	It remains to show that~$|S_i| \ge n_A - \delta_P^- + \ddp$ for all~$i \in \{0,\ldots,n_P-1\}$ follows from our assumptions.
	By assumption of the lemma we have~$n_P \le \delta_P^- - \ddp$.
	Hence, $|S_0| = n_A - n_P \ge n_A - \delta_P^- + \ddp$.
	We next show that~$|S_i| \ge |S_0|$ for all~$i \in [n_P-1]$.
	Observe that at the start we have~$\sum_{a \in S_{0}} \phi_0(a) = |S_0| \cdot \ddp$.
	Moreover, after the $i$th iteration of the loop in \cref{line:paper-loop} we have~$\sum_{a \in S_{i+1}} \phi_{i+1}(a) = \sum_{a \in S_{i}} \phi_{i}(a)$ as each paper gets~$\ddp$ reviews and the reviewer~$a$ in~$S_{i+1}\setminus S_i$ starts with~$\phi_{i+1}(a) = \ddp$. 
	Observe that~$\phi_i(a) \le \ddp$ for all~$a\in A$ and~$i \in \{0,\ldots,n_P-1\}$.
	Thus, we have~$|S_i| \ddp \ge \sum_{a \in S_{i}} \phi_i(a) = \sum_{a \in S_{0}} \phi_0(a) = |S_0| \ddp$ and, hence, $|S_i| \ge |S_0|$.
	This completes the correctness proof.
	
	As to the running time, everything outside the loop starting in \cref{line:paper-loop} clearly runs in linear time.
	As to the part inside the loop, note that by keeping just one array of length~$n_A$ we can store the values of~$\phi$ in linear time. 
	Moreover, the reviewers for~$p$ are selected arbitrarily from $R$, which is doable in~$|N^-(p,E_A)|$ time.
	Hence, the loop in \cref{line:paper-loop} can be processed in~$O(n_P + |E_A|)$ time.
	Thus, the overall algorithm runs in~$O(n_A + n_P + |E_A|)$, that is, linear time.
\end{proof}

For our next result we replace the completely cycle-free property of the resulting review assignment with $z$-cycle freeness.
This implies that the idea of constructing the review assignment along its topological ordering (as done by \cref{alg:greedy-dag}) cannot be employed.
Instead,  \cref{alg:greedy-short-cycle-free} constructs greedily a maximal $z$-cycle-free assignment and then extends the assignment by replacing one review assignment by two other assignments.
The argument behind the replacement strategy is an extension of the argument in \cref{alg:greedy-dag} that there are always enough reviewers to assign in \crefrange{line:assign-reviewers}{line:assign-reviewers-end}.

To keep our arguments simple we first consider the case that each agent reviews at most one paper and each paper requires one review. 
Moreover, as before, we are in the setting that each paper has one author and each agent authors at most one paper. 
Formally, we have the following.

\begin{proposition}\label{lem:short-cycle-free-simple-setting}
	If~$\Delta_A^- \le 1 = \delta_P^+ = \Delta_P^+ = \ccr = \ddp$, $n_A \ge n_P$, $\delta_A^+ > z$, $\delta_P^- > z$, and~$n_P \le \delta_A^+ + \delta_P^- - 2z$, then \cref{alg:greedy-short-cycle-free} computes a~$\ccr$-$\ddp$-valid $z$-cycle-free review assignment in polynomial time.
\end{proposition}

\begin{algorithm}[t!]
	\caption{Greedy algorithm to compute a $\ccr$-$\ddp$-valid $z$-cycle-free review assignment~$E'$.}
	\label{alg:greedy-short-cycle-free}
	$E' \gets \emptyset$\;
	\While(\label{line:while-loop}){$\exists p \in P\colon |N^-(p,E')| < \ddp$}
	{
		\If(\label{line:greedy-grow-matching}){$\exists (a,p) \in E_A \setminus E' \colon E' \cup \{(a,p)\}$ is $z$-cycle free and~$|N^+(a,E')| < \ccr$}
		{
			\tcc{Case 1: greedy assignment of reviews as long as no $z$-cycles are created:}
			$E' \gets E' \cup \{(a,p)\}$ \label{line:greedy-grow-matching-end}\;
		}
		\Else
		{
			\tcc{Case 2: replace one review assignment by two:}
			pick~$(a',p') \in E'$ and~$a \in A$ so that $|N^+(a,E')| < \ccr$, $(a',p),(a,p')\in E_A$, and~$(E' \setminus \{(a',p')\}) \cup \{(a',p),(a,p')\}$ is $z$-cycle free \label{line:swap-candidates}\;
			$E' \gets (E' \setminus \{(a',p')\}) \cup \{(a',p),(a,p')\}$ \label{line:swap-reviews-end}\;
		}
	}
	\textbf{return} $E'$ \label{line:return-E'}
\end{algorithm}

\begin{proof}
	Obviously, \cref{alg:greedy-short-cycle-free} terminates after at most~$n_P$ iterations of the while loop as in each iteration the number of assigned reviews increases.
	Moreover, a~$\ccr$-$\ddp$-valid $z$-cycle-free review assignment is returned if~$a,a',p'$ as described in case 2 (\cref{line:swap-candidates}) always exist.
	To prove their existence, we introduce some notation.
	For some~$v \in A \cupdot P$ let~$N_z^+(v,E' \cupdot E_P)$ be the $z$-out-neighborhood of~$v$, that is, the set of vertices that can be reached from~$v$ in the review graph~$(A \cupdot P, E' \cupdot E_P)$ via a path of length at least one and at most~$2z$.
	Similarly, let~$N_z^-(v,E' \cupdot E_P)$ be the $2z$-in-neighborhood of~$v$, that is, the set of vertices that can reach~$v$ in the review graph~$(A \cupdot P, E' \cupdot E_P)$ via a path of length at least one and at most~$2z$.
	Note that if~$v \in N_z^-(v,E' \cupdot E_P)$, then also~$v \in N_z^+(v,E' \cupdot E_P)$ and~$v$ is contained in a review cycle of length~$z$ (that is a directed cycle of length~$2z$ in $(A \cupdot P, E' \cupdot E_P)$). 
	Subsequently, we present upper bounds on the size of~$N_z^-(v,E' \cupdot E_P)$ and~$N_z^+(v,E' \cupdot E_P)$ for~$v \in A \cupdot P$ thereby proving the existence of~$a, a', p'$.
	 
	Let~$p \in P$ be the paper without reviewer selected in \cref{line:while-loop} when the algorithm enters case 2. 
	Let~$A_p \subseteq A$ be the set of agents that could review~$p$ without creating a $z$-cycle, that is, $A_p := \{a \in A \mid (a,p) \in E_A \land E' \cup \{(a,p)\}$ is $z$-cycle free$\}$.
	Since~$\ddp=\ccr = \Delta_P^+ = 1$, there are at most~$z$ agents whose assignment to review~$p$ would create a review cycle, that is, $|N_z^+(p,E' \cupdot E_P) \cap A| \le z$, and thus~$|A_p| \ge \delta_P^- - z$.
	Since we are in case 2, no more review assignments could be added without creating a~$z$-cycle. 
	Hence, the algorithm assigned the at least~$\delta_P^- - z$ potential reviewers in~$A_p$ to different papers. 
	Let~$P_p$ be the set of these papers.
	Since~$\ddp = \ccr = 1$ we have~$|P_p| = |A_p| \ge \delta_P^- - z$.
	
	Let~$a \in A$ be an arbitrary agent without assigned review, that is, $\nexists p''\colon (a,p'') \in E'$.
	Since~$\ddp=\ccr = \Delta_A^- = 1$, we have~$|N_z^-(a,E' \cupdot E_P) \cap P| \le z$. 
	Thus, there are~$\delta_A^+ - z$ papers that~$a$ could review without creating a $z$-cycle; let~$P_a$ denote the set of these papers.
	Since we assume that~$n_P \le \delta_A^+ + \delta_P^- - 2z$, it follows that there is a~$p' \in P_a \cap P_p$. 
	By definition of~$P_p$ there is an agent~$a'$ with~$(a',p') \in E'$ and~$a' \in A_p$.
	Thus, $a, a', p'$ exist and $E'$ can be updated to~$(E' \setminus \{(a',p')\}) \cup \{(a',p),(a,p')\}$~in~\cref{line:swap-reviews-end}.
\end{proof}

We now turn our attention to our general case where agents can author and review many papers and papers can have multiple authors and can require several reviews.
While the conditions that guarantee the existence of a $z$-cycle-free review assignment need adjustments, we can still use \cref{alg:greedy-short-cycle-free} together with a correctness proof that follows a similar pattern as the proof of \cref{lem:short-cycle-free-simple-setting}.

\begin{restatable}{theorem}{solExists}
	\label{thm:sol-exists}
	If, $n_A \cdot \ccr \ge n_P \cdot \ddp$, $\delta_A^+ > 2(\Delta_A^- \cdot \ddp)^{z} + \ccr$, $\delta_P^- > 2(\Delta_P^+ \cdot \ccr)^{z}+\ddp$, and~$n_P \le \delta_A^+ - 2(\Delta_A^- \cdot \ddp)^{z} - \ccr + (\ccr / \ddp) (\delta_P^- - 2(\Delta_P^+ \cdot \ccr)^{z} - \ddp)$, then \cref{alg:greedy-short-cycle-free} computes a~$\ccr$-$\ddp$-valid $z$-cycle-free review assignment in polynomial time.
\end{restatable}
\begin{proof}
	We use the same notation as in the proof of \cref{lem:short-cycle-free-simple-setting} and similarly to this proof we need to show that~$a,a',p'$ as described in \cref{line:swap-candidates} actually always exist.

	Let~$p \in P$ be the paper with a missing review selected in \cref{line:while-loop} and the algorithm entered case 2. 
	Let~$A_p \subseteq P$ be the set of agents that could review~$p$ without creating a $z$-cycle, that is, $A_p := \{a \in A \mid (a,p) \in E_A \setminus E' \land E' \cup \{(a,p)\}$ is $z$-cycle free$\}$.
	As every paper has at most~$\Delta_P^+$ authors and every author has at most~$\ccr$ assigned papers to review, it follows that
	\begin{align*}
		{}& |N_z^+(p,E' \cupdot E_P) \cap A| 	\le \Delta_P^+ \cdot \sum_{i = 0}^{z-1} (\Delta_P^+ \cdot \ccr)^{i} \\
		={}& \Delta_P^+ \cdot ((\Delta_P^+ \cdot \ccr)^{z}-1)/((\Delta_P^+ \cdot \ccr)-1)\\
		<{}& 2(\Delta_P^+ \cdot \ccr)^{z}.
	\end{align*}
	Thus,~$|A_p| > \delta_P^- - 2(\Delta_P^+ \cdot \ccr)^{z} - \ddp$, as at most~$\ddp-1$ agents are already assigned to~$p$ and at most $|N_z^+(p,E' \cupdot E_P) \cap A|<2(\Delta_P^+ \cdot \ccr)^{z}$ agents cannot review $p$ because this would cause a review cycle of length at most $z$.
	When case 2 was entered, no more review assignment could be added without creating a~$z$-cycle. 
	Hence, the algorithm assigned the potential reviewers in~$A_p$ already to different papers.
	Let~$P_p$ be the set of these papers.
	Note that~$|P_p| \ge \ccr |A_p| / \ddp$.
	
	Let~$a \in A$ be an arbitrary agent that can do one more review, that is, $|N^+(a,E')|< \ccr$. 
	Using a similar argument as above, we can show~$|N_z^-(a,E' \cupdot E_P) \cap P| < 2(\Delta_A^- \cdot \ddp)^{z}$. 
	Thus, there are more than~$\delta_A^+ - 2(\Delta_A^- \cdot \ddp)^{z} - \ccr$ papers that~$a$ could review additionally without creating a $z$-cycle; let~$P_a$ denote the set of these papers.
	Note that by our assumptions that $\delta_A^+ > 2(\Delta_A^- \cdot \ddp)^{z} + \ccr$ and $\delta_P^- > 2(\Delta_P^+ \cdot \ccr)^{z}+\ddp$, $P_a$ and $P_p$ are both non-empty.
	Since~$n_P \le \delta_A^+ - 2(\Delta_A^- \cdot \ddp)^{z} - \ccr + (\ccr / \ddp) (\delta_P^- - 2(\Delta_P^+ \cdot \ccr)^{z} - \ddp) < |P_a| + |P_p|$, it follows that there is a~$p' \in P_a \cap P_p$.
	By definition of~$P_p$ there is an agent~$a'$ with~$(a',p') \in E'$ and~$a' \in A_p$.
	Thus, $a, a', p'$ exist and $E'$ can be updated to~$(E' \setminus \{(a',p')\}) \cup \{(a',p),(a,p')\}$ in \cref{line:swap-reviews-end}.
	This finishes the correctness proof.
\end{proof}

To simplify the statement of \cref{thm:sol-exists} consider a ``symmetric'' case where~$n_A \ge n_P$, $\delta_P^- = \delta_A^+$, and~$\Delta_P^+ = \Delta_A^-$.
For brevity, set~$n := n_P$, $\delta := \delta_P^-$, and~$\Delta := \Delta_P^+$.
Let~$\coi$ be the maximum number of papers any agent is not qualified to review/has a conflict of interest with, that is, $\coi = n - \delta$.
Setting $\ccr = 6$ and~$\ddp = 3$ as in our experiments~we~get:

\begin{corollary}\label{cor:translatePcases}
	If~$n - 6 \ge 1.5 \cdot \coi + \Delta^z(6^z \cdot 2 + 3^z)$, then there always exists a $6$-$3$-valid $z$-cycle-free review assignment that can be found in polynomial time.
\end{corollary}

Considering that AAAI'22 had 9,251 submissions and that there was a submission limit of $10$ papers per author and assuming that each paper has at most ten authors (implying that~$\Delta = 10$) and that each author has at most 700 conflict of interests,
it follows that there is a $6$-$3$-valid $2$-cycle-free review assignment computable with \cref{alg:greedy-short-cycle-free}.

As we see in the experiments in the next section, our algorithm returns $2/3/4$-cycle-free review assignments even well beyond the theoretical guarantees given above.
We also remark that \cref{alg:greedy-short-cycle-free} allows for an easy extension to the weighted case which we use in our experiments in the next section.
To this end, in the first case (\cref{line:greedy-grow-matching-end}) we do not pick an arbitrary edge~$(a,p)$ but a eligible edge of maximum weight to be added to the assignment~$E'$.

\section{Experiments} \label{se:experiments}
    
In this section, we compare the weight of review assignments computed using different methods and analyze the occurrences of review cycles.\footnote{The code for our experiments is available at \url{github.com/n-boehmer/Combating-Collusion-Rings-is-Hard-but-Possible}.}
For this, we use a dataset from  the 2018 International Conference on Learning Representations (ICLR '18) prepared by \citet{DBLP:conf/ijcai/Xu0SS19}. 
\citet{DBLP:conf/ijcai/Xu0SS19} collected all $911$ papers submitted to ICLR '18 and the identity of all $2428$ authors. 
As reviewers identities are unknown, they considered all authors to be reviewers and computed for each author-paper pair a similarity score.\footnote{To the best of our knowledge, in all other publicly available datasets, there are similarity scores for reviewer-paper pairs but the link between the identities of authors and reviewers is missing (as this is considered sensitive information).} 

From the dataset of \citet{DBLP:conf/ijcai/Xu0SS19}, we created multiple instances of \wcfpr as follows. 
Given a number $n_P$ of papers and a ratio $r_{AP}$ of the numbers of agents and papers, we sample a subset of $n_P$ of the $911$ ICLR '18 papers and set this as our set of papers. 
Subsequently, we compute the set of all authors of one of these papers and sample a subset of $r_{AP}\cdot n_P$ authors and set this as our set of agents. 
Notably, the created instances can be seen as particularly challenging when it comes to avoiding review cycles, as in reality also ``uncritical'' reviewers, i.e., reviewers that do not author any paper, exist. 

As done in other papers using the same dataset, we focus on the case with $\ddp=3$ and $\ccr=6$, i.e., every paper needs exactly three reviews and each agent can review at most six papers \citep{DBLP:conf/ijcai/Xu0SS19,DBLP:conf/nips/JecmenZLSCF20}.  
We consider three different types of review assignments:
As ``optimal'' we denote a maximum-weight $\ccr$-$\ddp$-valid review assignment. 
Such an assignment can be computed using a simple Linear Program (LP) as, for instance, described by \citet{taylor2008optimal}. 
As ``optimal $z$-cycle free'' we denote a maximum-weight $\ccr$-$\ddp$-valid $z$-cycle-free review assignment. 
This solution can be computed by treating the LP of \citet{taylor2008optimal} as an Integer Linear Program (ILP) and adding for each possible $i$-cycle for $i\in [z]$ a separate constraint which imposes that at least one of the agent-paper pairs from the cycle is not assigned. 
We solved all (I)LPs using \citet{gurobi}.
Lastly, as ``heuristic $z$-cycle free'', we denote a $\ccr$-$\ddp$-valid $z$-cycle-free review assignment computed by the weighted variant of \cref{alg:greedy-short-cycle-free} as described at the end of \Cref{sec:heuristic}.\footnote{We could not use the heuristics of \citet{DBLP:conf/iccnc/Guo0CWL18} as these are not available and their algorithm details are ambiguous.} 
In all experiments conducted in this section, the heuristic always returned a solution despite the fact that most of the time we are beyond the setting in which \cref{thm:sol-exists} guarantees this behavior of the heuristic.
In experiment I presented in the following subsection, for $z=2/3/4$, an unoptimized implementation of our heuristic was always able to find a $z$-cycle-free review assignment in less than $30$ seconds, being on average around $2$ times faster than the ``optimal'' LP, on average around $3.7$ times faster than the ``optimal $2$-cycle free'' ILP, and on average more than $100$ times faster than the ``optimal $3$-cycle free'' ILP.

\subsection{Experiment I}
In this experiment, we focus on the case where the total number of needed reviews is the same as the total number of reviews that can be written,
which is in some sense the most ``challenging'' but probably also one of the more realistic scenarios.
Specifically, for $n_P\in \{150,175,\dots, 900\}$, we prepared $100$ instances with $r_{AP}=0.5$ as described above and computed for each of these instances the optimal, heuristic $2$/$3$/$4$-cycle-free, and optimal $2$-cycle-free review assignment. 
Moreover, for all instances with $n_P\leq 225$, we also computed the optimal $3$-cycle-free review assignment
(for larger instances the ILP solver run out of memory.)

To measure the ``price of $z$-cycle freeness'', in \Cref{fig:quality}, we display the weights of different  cycle-free review assignments divided by the weight of an optimal review assignment.  
What stands out here is that forbidding the existence of $2$-cycles only comes at the cost of decreasing the assignment's weight by on average at most $0.8\%$ (if the optimal $2$-cycle-free assignment is used). 
Turning to the results produced by our heuristic, the quality decrease for $2$/$3$/$4$-cycle-free assignments lies, on average, around $3.1\%$, $3.2\%$, and $3.3\%$. 
The weight of assignments computed using our heuristic is thus clearly worse than the weight of the optimal cycle-free assignment, yet still not far away from the the weight of an optimal assignment. 
What is particularly surprising here is that for both our heuristic and the optimal cycle-free assignment, whether $2$, $3$ or $4$ cycles are forbidden seems to be rather irrelevant for the quality decrease. 
All in all, it is encouraging that $2$/$3$/$4$-cycle freeness can be realized at a low cost independent of whether our heuristic or an ILP is used.

\begin{figure*}[t!]
	\centering
	\begin{minipage}[t]{.32\textwidth}
		\centering
		\resizebox{\textwidth}{!}{
\begin{tikzpicture}

\definecolor{color0}{rgb}{1,0.549019607843137,0}
\definecolor{color1}{rgb}{0.133333333333333,0.545098039215686,0.133333333333333}

\begin{axis}[
legend columns=2, 
legend cell align={left},
legend style={
	fill opacity=0.8,
	draw opacity=1,
	draw=none,
	text opacity=1,
	at={(0.5,1.3)},
	line width=3pt,
	anchor=north,
	/tikz/column 2/.style={
		column sep=10pt,
	}
},
legend entries={2-cycle free,
	heuristic cycle free,
	3-cycle free, optimal cycle free,
	4-cycle free},
tick align=outside,
tick pos=left,
x grid style={white!69.0196078431373!black},
xlabel={number of papers},
xmin=130, xmax=920,
xtick style={color=black},
y grid style={white!69.0196078431373!black},
ylabel={weight normalized by weight of optimal solution},
ymin=0.96, ymax=0.996,
ytick style={color=black},
ytick={0.96,0.965,0.97,0.975,0.98,0.985,0.99,0.995},
yticklabels={0.96,0.965,0.97,0.975,0.98,0.985,0.99,0.995},
ylabel style={yshift=0.25cm}
]
\addlegendimage{red!54.5098039215686!black}
\addlegendimage{gray}
\addlegendimage{color0}
\addlegendimage{gray,dash pattern=on 7pt off 4pt}
\addlegendimage{color1}
\addplot [line width=1.5pt, red!54.5098039215686!black, dash pattern=on 6pt off 3pt]
table {%
150 0.991926574347705
175 0.992667667406678
200 0.99291617740913
225 0.993643537526456
250 0.99334260733802
275 0.993225577072392
300 0.993021347245257
325 0.993346331176092
350 0.993463865218339
375 0.993775776392882
400 0.993766108819567
425 0.993673085374533
450 0.993939230465521
475 0.993977685215255
500 0.993248477116723
525 0.99404592781744
550 0.993685268834669
575 0.993997531834272
600 0.993819968360993
625 0.993811463542213
650 0.99409851273386
675 0.99407746489986
700 0.993974596195344
725 0.994040546507579
750 0.99407317116466
775 0.994204715548222
800 0.99418249296385
825 0.993983181584122
850 0.994331382646615
875 0.994149491000736
900 0.994207326081006
};
\addplot [line width=1.5pt, red!54.5098039215686!black, solid]
table {%
150 0.968510036128565
175 0.970098946584062
200 0.969364391288481
225 0.969438959252551
250 0.968174771579804
275 0.969704760187726
300 0.970073660852376
325 0.970128813155371
350 0.970676219132571
375 0.970174256617788
400 0.970220823462127
425 0.970061122471152
450 0.970227743264512
475 0.969908261794796
500 0.969815917985764
525 0.969998125197392
550 0.969896572554831
575 0.970086284368001
600 0.969873811465109
625 0.970483332012341
650 0.970187358458012
675 0.969976379222872
700 0.969855275388839
725 0.969835996796407
750 0.970122835807463
775 0.970456930723669
800 0.969805332582258
825 0.969737566702351
850 0.969929994333001
875 0.96995286974093
900 0.970152100435763
};
\addplot [line width=1.5pt, color0, dash pattern=on 6pt off 3pt]
table {%
150 0.990095207834939
175 0.990395627786302
200 0.991215902612412
225 0.991248268040838
};
\addplot [line width=1.5pt, color0, solid]
table {%
150 0.966043473776896
175 0.968108952737625
200 0.967259062577482
225 0.9675798599254
250 0.966406986141914
275 0.967920644516303
300 0.968627310900345
325 0.96850332978838
350 0.969269087073347
375 0.968729446810716
400 0.968811114708596
425 0.968611359847732
450 0.968941075892002
475 0.968496207309294
500 0.968634090803401
525 0.968692457779898
550 0.968805122684117
575 0.969009997310226
600 0.968635277974614
625 0.969274791882022
650 0.96894729877903
675 0.968647494593104
700 0.968819804355576
725 0.968811209687376
750 0.968936672999162
775 0.969334800692915
800 0.968637829782921
825 0.968667838687287
850 0.968774356243043
875 0.968853143574644
900 0.96891980682276
};
\addplot [line width=1.5pt, color1, solid]
table {%
150 0.963155572536885
175 0.965811251516701
200 0.964994405771822
225 0.965464534247496
250 0.964522259464502
275 0.966025804007046
300 0.967073691097973
325 0.966936040273682
350 0.967795516455457
375 0.967319003698681
400 0.967396122108348
425 0.967156483992846
450 0.967584863015462
475 0.967388147482651
500 0.967370155190196
525 0.967361876749796
550 0.9675590279648
575 0.967864802732198
600 0.967387378710376
625 0.968106638533988
650 0.967704720281115
675 0.967432414826679
700 0.967826624122599
725 0.967818703940084
750 0.967805520032026
775 0.96830520830651
800 0.967578463359523
825 0.967599746775323
850 0.967680763851255
875 0.967889007667123
900 0.967937269001565
};
\end{axis}

\end{tikzpicture}}
		\caption{For different values of $z$, weight of an optimal/heuristic $z$-cycle-free assignment divided by the weight of an optimal assignment.}\label{fig:quality}
	\end{minipage}\hfill
	\begin{minipage}[t]{.32\textwidth}
		\centering
		\resizebox{\textwidth}{!}{
\begin{tikzpicture}

\definecolor{color0}{rgb}{1,0.549019607843137,0}
\definecolor{color1}{rgb}{0.133333333333333,0.545098039215686,0.133333333333333}

\begin{axis}[
legend columns=2, 
legend cell align={left},
legend style={
	fill opacity=0.8,
	draw opacity=1,
	draw=none,
	text opacity=1,
	at={(0.5,1.3)},
	line width=3pt,
	anchor=north,
	/tikz/column 2/.style={
		column sep=10pt,
	}
},
legend entries={2-cycles, optimal,
	2/3-cycles, heuristic 2-cycle free,
	2/3/4-cycles, heuristic 3-cyclce free},
tick align=outside,
tick pos=left,
x grid style={white!69.0196078431373!black},
xlabel={number of papers},
xmin=130, xmax=920,
xtick style={color=black},
y grid style={white!69.0196078431373!black},
ylabel={fraction of agents in a review cycle},
ymin=0.15, ymax=0.826666666666667,
ytick style={color=black},
ytick={-0.1,0,0.1,0.2,0.3,0.4,0.5,0.6,0.7,0.8,0.9},
yticklabels={−0.1,0.0,0.1,0.2,0.3,0.4,0.5,0.6,0.7,0.8,0.9}
]
\addlegendimage{red!54.5098039215686!black}
\addlegendimage{gray}
\addlegendimage{color0}
\addlegendimage{gray,dash pattern=on 2pt off 2pt}
\addlegendimage{color1}
\addlegendimage{gray,dash pattern=on 7pt off 4pt}
\addplot [line width=1.5pt, red!54.5098039215686!black,solid]
table {%
150 0.404210526315789
175 0.374318181818182
200 0.378415841584159
225 0.374690265486726
250 0.388571428571429
275 0.361884057971015
300 0.346887417218543
325 0.358036809815951
350 0.342386363636364
375 0.348085106382979
400 0.341293532338309
425 0.34
450 0.328318584070796
475 0.331344537815126
500 0.327171314741036
525 0.337034220532319
550 0.327391304347826
575 0.326666666666667
600 0.323654485049834
625 0.327220447284345
650 0.328957055214724
675 0.326272189349112
700 0.319202279202279
725 0.324573002754821
750 0.328031914893617
775 0.327731958762887
800 0.330573566084788
825 0.327893462469734
850 0.328544600938967
875 0.321461187214612
900 0.323237250554324
};
\addplot [line width=1.5pt, color0,solid]
table {%
150 0.581052631578947
175 0.544545454545455
200 0.525346534653466
225 0.507079646017699
250 0.523174603174603
275 0.495072463768116
300 0.48476821192053
325 0.477668711656442
350 0.468068181818182
375 0.46468085106383
400 0.451641791044776
425 0.453521126760563
450 0.440353982300885
475 0.440168067226891
500 0.429800796812749
525 0.439391634980989
550 0.430652173913043
575 0.429861111111111
600 0.425913621262458
625 0.423514376996805
650 0.430122699386503
675 0.42491124260355
700 0.416695156695157
725 0.422038567493113
750 0.426648936170213
775 0.420670103092784
800 0.428029925187032
825 0.426004842615012
850 0.419389671361502
875 0.412328767123288
900 0.413968957871397
};
\addplot [line width=1.5pt, color1,solid]
table {%
150 0.755263157894737
175 0.720909090909091
200 0.701584158415842
225 0.683362831858407
250 0.674603174603174
275 0.656086956521739
300 0.648476821192053
325 0.642208588957055
350 0.640454545454545
375 0.620106382978724
400 0.602288557213931
425 0.610892018779343
450 0.593362831858407
475 0.593781512605042
500 0.578406374501992
525 0.583650190114068
550 0.582028985507246
575 0.565625
600 0.570232558139535
625 0.563833865814696
650 0.57521472392638
675 0.567337278106509
700 0.551509971509971
725 0.56
750 0.569840425531915
775 0.549072164948454
800 0.564638403990025
825 0.562905569007264
850 0.551596244131455
875 0.546529680365297
900 0.545454545454546
};
\addplot [line width=1.5pt, color0, dotted]
table {%
150 0.374736842105263
175 0.341136363636364
200 0.311089108910891
225 0.310973451327434
250 0.304761904761905
275 0.29
300 0.265562913907285
325 0.259018404907975
350 0.242272727272727
375 0.236702127659575
400 0.229353233830846
425 0.234835680751174
450 0.229380530973451
475 0.222689075630252
500 0.223187250996016
525 0.224562737642586
550 0.209420289855072
575 0.205902777777778
600 0.216146179401993
625 0.204600638977636
650 0.20361963190184
675 0.216272189349112
700 0.205242165242165
725 0.1966391184573
750 0.202287234042553
775 0.199536082474227
800 0.205586034912718
825 0.205859564164649
850 0.199765258215962
875 0.195799086757991
900 0.197073170731707
};
\addplot [line width=1.5pt, color1, dotted]
table {%
150 0.679473684210526
175 0.654318181818182
200 0.631683168316832
225 0.608849557522124
250 0.591428571428572
275 0.571159420289855
300 0.554834437086093
325 0.539018404907976
350 0.515568181818182
375 0.503617021276596
400 0.494228855721393
425 0.501784037558686
450 0.47929203539823
475 0.465546218487395
500 0.470836653386454
525 0.476197718631179
550 0.463115942028985
575 0.455347222222222
600 0.4578073089701
625 0.45150159744409
650 0.449263803680982
675 0.450532544378698
700 0.430598290598291
725 0.426391184573003
750 0.439414893617021
775 0.41201030927835
800 0.430324189526185
825 0.438740920096852
850 0.423192488262911
875 0.424292237442922
900 0.41840354767184
};
\addplot [line width=1.5pt, color1, dash pattern=on 6pt off 3pt]
table {%
150 0.573684210526316
175 0.526136363636364
200 0.527524752475248
225 0.489203539823009
250 0.471587301587302
275 0.459565217391304
300 0.442384105960265
325 0.430552147239264
350 0.414431818181818
375 0.398936170212766
400 0.379303482587065
425 0.38037558685446
450 0.38070796460177
475 0.358235294117647
500 0.358884462151394
525 0.364334600760456
550 0.34963768115942
575 0.33875
600 0.348172757475083
625 0.340830670926518
650 0.348159509202454
675 0.339349112426035
700 0.316581196581196
725 0.314490358126722
750 0.330531914893617
775 0.300567010309278
800 0.317605985037406
825 0.322905569007264
850 0.316431924882629
875 0.30310502283105
900 0.300709534368071
};
\end{axis}

\end{tikzpicture}} 
		\caption{Fraction of agents that are part of a review cycle of at most some length for different types of assignment.}\label{fig:cyles}
	\end{minipage}\hfill
	\begin{minipage}[t]{.32\textwidth}
	\centering
		\resizebox{\textwidth}{!}{
\begin{tikzpicture}

\definecolor{color0}{rgb}{1,0.549019607843137,0}
\definecolor{color1}{rgb}{0.133333333333333,0.545098039215686,0.133333333333333}

\begin{axis}[
legend columns=2, 
legend cell align={left},
legend style={
	fill opacity=0.8,
	draw opacity=1,
	draw=none,
	text opacity=1,
	at={(0.5,1.3)},
	line width=3pt,
	anchor=north,
	/tikz/column 2/.style={
		column sep=10pt,
	}
},
legend entries={2-cycles, optimal,
	2/3-cycles, heuristic 2-cycle free,
	2/3/4-cycles, heuristic 3-cyclce free},
tick align=outside,
tick pos=left,
x grid style={white!69.0196078431373!black},
xlabel={number of papers},
xmin=130, xmax=920,
xtick style={color=black},
y grid style={white!69.0196078431373!black},
ylabel={fraction of papers in a review cycle},
ymin=0.08, ymax=0.35,
ytick style={color=black},
ytick={-0.1,0,0.1,0.15,0.2,0.25,0.3,0.4,0.5,0.6,0.7,0.8,0.9},
yticklabels={−0.1,0.0,0.1,0.15,0.2,0.25,0.3,0.4,0.5,0.6,0.7,0.8,0.9}
]
\addlegendimage{red!54.5098039215686!black}
\addlegendimage{gray}
\addlegendimage{color0}
\addlegendimage{gray,dotted}
\addlegendimage{color1}
\addlegendimage{gray,dash pattern=on 6pt off 3pt}
\addplot [line width=1.5pt, red!54.5098039215686!black,solid]
table {%
150 0.182933333333333
175 0.169828571428571
200 0.1709
225 0.168711111111111
250 0.17552
275 0.164581818181818
300 0.1618
325 0.163815384615385
350 0.156971428571429
375 0.161706666666667
400 0.1553
425 0.155011764705882
450 0.154355555555556
475 0.154063157894737
500 0.15072
525 0.155885714285714
550 0.151454545454545
575 0.150678260869565
600 0.151066666666667
625 0.152864
650 0.152984615384615
675 0.152562962962963
700 0.148228571428571
725 0.150813793103448
750 0.153173333333333
775 0.151354838709677
800 0.154975
825 0.154060606060606
850 0.152658823529412
875 0.14928
900 0.153466666666667
};
\addplot [line width=1.5pt, color0,solid]
table {%
150 0.260133333333333
175 0.246971428571429
200 0.2384
225 0.230755555555556
250 0.23736
275 0.227636363636364
300 0.2242
325 0.219569230769231
350 0.216057142857143
375 0.215306666666667
400 0.2089
425 0.209176470588235
450 0.207822222222222
475 0.205178947368421
500 0.20036
525 0.205942857142857
550 0.202218181818182
575 0.202330434782609
600 0.201
625 0.200096
650 0.204646153846154
675 0.203081481481482
700 0.197457142857143
725 0.200220689655173
750 0.202613333333333
775 0.198270967741936
800 0.20495
825 0.203951515151515
850 0.199411764705882
875 0.19632
900 0.201577777777778
};
\addplot [line width=1.5pt, color1,solid]
table {%
150 0.334
175 0.3192
200 0.3143
225 0.306844444444444
250 0.3036
275 0.299636363636363
300 0.297333333333333
325 0.293107692307692
350 0.293714285714286
375 0.285013333333333
400 0.27855
425 0.283529411764706
450 0.279466666666667
475 0.277136842105263
500 0.27208
525 0.274819047619048
550 0.275636363636364
575 0.266469565217391
600 0.270766666666667
625 0.269632
650 0.275538461538462
675 0.274281481481481
700 0.263257142857143
725 0.269351724137931
750 0.273013333333333
775 0.262890322580645
800 0.273375
825 0.273963636363636
850 0.266423529411765
875 0.263908571428571
900 0.270755555555556
};
\addplot [line width=1.5pt, color0, dotted]
table {%
150 0.1748
175 0.158171428571429
200 0.1467
225 0.145155555555556
250 0.1436
275 0.137527272727273
300 0.1262
325 0.121415384615385
350 0.115428571428571
375 0.11312
400 0.1099
425 0.113223529411765
450 0.109555555555556
475 0.107621052631579
500 0.10684
525 0.108342857142857
550 0.102072727272727
575 0.0987826086956522
600 0.104466666666667
625 0.099424
650 0.0994461538461539
675 0.106814814814815
700 0.100057142857143
725 0.0943724137931035
750 0.09808
775 0.0965161290322581
800 0.101
825 0.101066666666667
850 0.0977882352941176
875 0.0949257142857143
900 0.0976222222222222
};
\addplot [line width=1.5pt, color1, dotted]
table {%
150 0.3
175 0.292457142857143
200 0.2861
225 0.276266666666667
250 0.27024
275 0.264436363636364
300 0.256933333333333
325 0.248738461538461
350 0.240228571428571
375 0.235733333333333
400 0.2316
425 0.238447058823529
450 0.227066666666667
475 0.221052631578947
500 0.22376
525 0.227047619047619
550 0.223890909090909
575 0.218539130434783
600 0.221166666666667
625 0.218144
650 0.2204
675 0.220533333333333
700 0.210142857142857
725 0.206951724137931
750 0.213946666666667
775 0.202477419354839
800 0.214
825 0.216484848484849
850 0.209882352941177
875 0.20864
900 0.210311111111111
};
\addplot [line width=1.5pt, color1, dash pattern=on 6pt off 3pt]
table {%
150 0.259066666666667
175 0.236914285714286
200 0.2411
225 0.225422222222222
250 0.21872
275 0.213163636363636
300 0.208133333333333
325 0.201661538461539
350 0.193371428571429
375 0.187146666666667
400 0.1797
425 0.182164705882353
450 0.182
475 0.171621052631579
500 0.1724
525 0.174895238095238
550 0.169636363636364
575 0.163304347826087
600 0.169933333333333
625 0.164352
650 0.171384615384615
675 0.16637037037037
700 0.154657142857143
725 0.154510344827586
750 0.16104
775 0.147406451612903
800 0.157475
825 0.160436363636364
850 0.157458823529412
875 0.149371428571429
900 0.150511111111111
};
\end{axis}

\end{tikzpicture}} 
	\caption{Fraction of papers that are part of a review cycle of at most some length for different types of assignment.} \label{fig:fracPapers}
	\end{minipage}
\end{figure*}

The necessity of dealing with review cycles is underlined by the data displayed in \Cref{fig:cyles}. 
Here, we show the fraction of agents that are contained in at least one review cycle of some length in an optimal assignment and in a heuristic $2/3$-cycle-free assignment. 
Overall, as the number of papers increases the fraction of agents contained in review cycles constantly decreases, yet for all considered values of $n_P$ the results are worrisome. 
In the optimal assignment for $150$ papers, the fraction of agents contained in a review cycle of length at most $2/3/4$ is, on average, $40\%/58\%/76\%$ , while even for $900$ papers, still $32\%/41\%/55\%$ of agents are contained in a review cycle. 
Considering heuristic $z$-cycle-free assignments, the fraction of agents contained in a cycle of length $z+1$ is considerably lower than for the optimal solution but still non-negligible (the results for optimal $2/3$-cycle-free assignments are similar to the displayed results for our heuristic).
    
We also computed the fraction of papers that are contained in at least one review cycle (see \Cref{fig:fracPapers}).  
The results are as in \Cref{fig:cyles} with all values roughly halved, e.g, even in the optimal assignment for $900$ papers, $15\%/20\%/27\%$ of papers are contained in a review cycle of length at most $2/3/4$. 
An intuitive explanation for this difference between agents and papers is that the number of papers is twice the number of agents and that there exist some papers without reviewing authors. 
Overall, it is striking that even for a high number of papers, in an optimal assignment around $15\%$ of papers could have a considerably higher chance of getting accepted if two agents coordinate to give each others paper better reviews and $32\%$ of reviewers would have an opportunity to participate in such a collusion. 

  \begin{figure*}[t!]
\begin{minipage}[t]{.48\textwidth}
		\centering
	\resizebox{0.7\textwidth}{!}{
\begin{tikzpicture}

\definecolor{color0}{rgb}{1,0.549019607843137,0}
\definecolor{color1}{rgb}{0.133333333333333,0.545098039215686,0.133333333333333}

\begin{axis}[
legend columns=2, 
legend cell align={left},
legend style={
	fill opacity=0.8,
	draw opacity=1,
	draw=none,
	text opacity=1,
	at={(0.5,1.3)},
	line width=3pt,
	anchor=north,
	/tikz/column 2/.style={
		column sep=10pt,
	}
},
legend entries={2-cycle free,
	heuristic cycle free,
	3-cycle free, optimal cycle free,
	4-cycle free},
tick align=outside,
tick pos=left,
x grid style={white!69.0196078431373!black},
xlabel={$\nicefrac{\text{number of agents}}{\text{number of papers}}$},
xmin=0.425, xmax=2.075,
xtick style={color=black},
xtick={0.4,0.6,0.8,1,1.2,1.4,1.6,1.8,2,2.2},
xticklabels={0.4,0.6,0.8,1.0,1.2,1.4,1.6,1.8,2.0,2.2},
y grid style={white!69.0196078431373!black},
ylabel={weight normalized by weight of optimal solution},
ymin=0.96, ymax=0.996,
ytick style={color=black},
ytick={0.96,0.965,0.97,0.975,0.98,0.985,0.99,0.995},
yticklabels={0.96,0.965,0.97,0.975,0.98,0.985,0.99,0.995},
ylabel style={yshift=0.25cm}
]
\addlegendimage{red!54.5098039215686!black}
\addlegendimage{gray}
\addlegendimage{color0}
\addlegendimage{gray,dotted}
\addlegendimage{color1}
\addplot [line width=1.5pt, red!54.5098039215686!black, dashed]
table {%
0.5 0.99299232792165
0.6 0.992139334565611
0.7 0.991215648799322
0.8 0.99092817572984
0.9 0.989888835807197
1 0.989737930928185
1.1 0.988587875567876
1.2 0.988763409126611
1.3 0.988823725889548
1.4 0.988925039521304
1.5 0.988014789797783
1.6 0.988213955685074
1.7 0.987298334469413
1.8 0.987264589777595
1.9 0.987099420216515
2 0.987422640839029
};
\addplot [line width=1.5pt, red!54.5098039215686!black, solid]
table {%
0.5 0.970185966754706
0.6 0.972513409102817
0.7 0.974816137800935
0.8 0.976451310543773
0.9 0.976991773740888
1 0.978140595363959
1.1 0.978432905015212
1.2 0.97919904290603
1.3 0.980238443691698
1.4 0.980809690906076
1.5 0.980684375852006
1.6 0.981299712490017
1.7 0.980941033315102
1.8 0.981039809073676
1.9 0.981004854411886
2 0.98171819039823
};
\addplot [line width=1.5pt, color0, dashed]
table {%
0.5 0.990100190448635
0.6 0.989128957692013
0.7 0.986468054265244
0.8 0.987259605540806
0.9 0.98491216585713
1 0.986173439919578
1.1 0.986254303727512
1.2 0.985185220050601
1.3 0.984613392020042
1.4 0.986032188267647
1.5 0.983674858630122
1.6 0.984519597035013
1.7 0.984462465907539
1.8 0.986514543940129
1.9 0.98557797182043
2 0.98483021478865
};
\addplot [line width=1.5pt, color0, solid]
table {%
0.5 0.968222631976865
0.6 0.970248390395589
0.7 0.972423288077584
0.8 0.974080237503292
0.9 0.974367559399936
1 0.975554486059007
1.1 0.975813151616617
1.2 0.976506001937259
1.3 0.977403263692971
1.4 0.978225979592959
1.5 0.97772185696346
1.6 0.978467137102889
1.7 0.977953037659628
1.8 0.977987054254864
1.9 0.977976512637854
2 0.9786844403342
};
\addplot [line width=1.5pt, color1, solid]
table {%
0.5 0.965856734019942
0.6 0.967670415112128
0.7 0.969843674884357
0.8 0.971667928495624
0.9 0.971848646975308
1 0.973188253845329
1.1 0.973448219189592
1.2 0.974056772317046
1.3 0.975004978820614
1.4 0.975903684536758
1.5 0.975256965220739
1.6 0.976023222184217
1.7 0.97564529218127
1.8 0.975532125449385
1.9 0.975703353905242
2 0.976270513989353
};
\end{axis}

\end{tikzpicture}}  
	\caption{For different values of $z$, weight of a optimal/heuristic $z$-cycle-free assignment divided by the weight of an optimal assignment.} \label{fig:qualPapers}
	\end{minipage} \hfill
	\begin{minipage}[t]{.48\textwidth}
		\centering
		\resizebox{0.7\textwidth}{!}{
\begin{tikzpicture}

\definecolor{color0}{rgb}{1,0.549019607843137,0}
\definecolor{color1}{rgb}{0.133333333333333,0.545098039215686,0.133333333333333}

\begin{axis}[
legend columns=2, 
legend cell align={left},
legend style={
	fill opacity=0.8,
	draw opacity=1,
	draw=none,
	text opacity=1,
	at={(0.5,1.3)},
	line width=3pt,
	anchor=north,
	/tikz/column 2/.style={
		column sep=10pt,
	}
},
legend entries={2-cycles, fraction agents,
	2/3-cycles,fraction papers,
	2/3/4-cycles},
tick align=outside,
tick pos=left,
x grid style={white!69.0196078431373!black},
xlabel={$\nicefrac{\text{number of agents}}{\text{number of papers}}$},
xmin=0.425, xmax=2.075,
xtick style={color=black},
xtick={0.4,0.6,0.8,1,1.2,1.4,1.6,1.8,2,2.2},
xticklabels={0.4,0.6,0.8,1.0,1.2,1.4,1.6,1.8,2.0,2.2},
y grid style={white!69.0196078431373!black},
ylabel={fraction contained in a review cylce},
ymin=0.14, ymax=0.71,
ytick style={color=black},
ytick={-0.1,0,0.1,0.2,0.3,0.4,0.5,0.6,0.7,0.8,0.9},
yticklabels={−0.1,0.0,0.1,0.2,0.3,0.4,0.5,0.6,0.7,0.8,0.9}
]
\addlegendimage{red!54.5098039215686!black}
\addlegendimage{gray}
\addlegendimage{color0}
\addlegendimage{gray,dash pattern=on 7pt off 4pt}
\addlegendimage{color1}
\addplot [line width=1.5pt, red!54.5098039215686!black]
table {%
0.5 0.38019801980198
0.6 0.352561983471074
0.7 0.330992907801418
0.8 0.306645962732919
0.9 0.281436464088398
1 0.271492537313433
1.1 0.259502262443439
1.2 0.244688796680498
1.3 0.235708812260536
1.4 0.219074733096085
1.5 0.211661129568106
1.6 0.20190031152648
1.7 0.200850439882698
1.8 0.191218836565097
1.9 0.188293963254593
2 0.176234413965087
};
\addplot [line width=1.5pt, color0]
table {%
0.5 0.526336633663367
0.6 0.47603305785124
0.7 0.434751773049645
0.8 0.4
0.9 0.358563535911603
1 0.343034825870647
1.1 0.325701357466063
1.2 0.302697095435685
1.3 0.289885057471264
1.4 0.269501779359431
1.5 0.25953488372093
1.6 0.247788161993769
1.7 0.242961876832845
1.8 0.229778393351801
1.9 0.22753280839895
2 0.215286783042394
};
\addplot [line width=1.5pt, color1]
table {%
0.5 0.684158415841585
0.6 0.615371900826446
0.7 0.549503546099291
0.8 0.49472049689441
0.9 0.442707182320442
1 0.42318407960199
1.1 0.388823529411765
1.2 0.364315352697095
1.3 0.347049808429119
1.4 0.32067615658363
1.5 0.310797342192691
1.6 0.296666666666667
1.7 0.283020527859238
1.8 0.270332409972299
1.9 0.264304461942257
2 0.253840399002494
};
\addplot [line width=1.5pt, red!54.5098039215686!black, dash pattern=on 6pt off 3pt]
table {%
0.5 0.17485
0.6 0.187
0.7 0.20625
0.8 0.2105
0.9 0.2135
1 0.22715
1.1 0.2339
1.2 0.23945
1.3 0.2445
1.4 0.2418
1.5 0.2474
1.6 0.2496
1.7 0.25885
1.8 0.25725
1.9 0.2669
2 0.2618
};
\addplot [line width=1.5pt, color0, dash pattern=on 6pt off 3pt]
table {%
0.5 0.2409
0.6 0.2527
0.7 0.26985
0.8 0.2731
0.9 0.27145
1 0.2858
1.1 0.2942
1.2 0.29375
1.3 0.2989
1.4 0.2953
1.5 0.30355
1.6 0.30495
1.7 0.313
1.8 0.3091
1.9 0.31985
2 0.3175
};
\addplot [line width=1.5pt, color1, dash pattern=on 6pt off 3pt]
table {%
0.5 0.30695
0.6 0.321
0.7 0.3357
0.8 0.33245
0.9 0.3331
1 0.34975
1.1 0.3482
1.2 0.35105
1.3 0.3538
1.4 0.3482
1.5 0.35985
1.6 0.3623
1.7 0.36045
1.8 0.3608
1.9 0.3668
2 0.36935
};
\end{axis}

\end{tikzpicture}} 
		\caption{Fraction of agents/papers that are part of a review cycle of at most some length in an optimal assignment for $200$ papers and between $100$ and $400$ agents.} \label{fig:proba_cyles}
	\end{minipage}
\end{figure*}
\subsection{Experiment II} 
In this experiment, we analyze how the results from experiment I depend on the assumption that the supply and demand of reviews exactly matches. 
In particular, as describe before, for $r_{AP}\in \{0.5,0.6,\dots,1.9,2\}$  we prepared $100$ instances with $200$ papers and $r_{AP}\cdot 200$ agents  (we also repeated this experiment for $400$ and $600$ papers producing similar results) and computed the different types of review assignments. 
Considering the assignment weights (see \Cref{fig:qualPapers}), increasing $r_{AP}$ from $0.5$ to $2$, the normalized weight of an optimal $2/3$-cycle-free assignment decreases by $0.005$ to $0.987/0.985$, while the normalized weight of a heuristic $2/3/4$-cycle-free assignment increases by $0.01$ to $0.982/0.979/0.976$:  
our heuristic performs particularly well if there are (considerably) more reviews available then needed; this supports our theoretical statements for our heuristic in \cref{sec:heuristic}. 

Turning to the possible impact of review cycles, we visualize the fraction of agents/papers contained in a review cycle in an optimal assignment in \Cref{fig:proba_cyles}.\footnote{For readability, we do not display the values for the optimal/heuristic cycle-free assignment, as their relationship to the optimal assignment is again similar as in \Cref{fig:cyles}.} 
While the fraction of agents contained in a review cycle constantly and significantly decreases if more and more agents are added, the fraction of papers contained in a cycle constantly increases. 
The former observation is quite intuitive, as when more and more agents are added, the average review load decreases and even if the number of review cycles remains the same, it is likely that the fraction of agents contained in one gets smaller. 
The latter observation is less intuitive but probably a consequence of the fact that, starting with $r_{AP}=0.5$,
for some papers none of the authors is part of the agent set, implying that these papers cannot be part of a review cycle; however, if we start to add more and more agents, more and more papers can potentially be part of a review cycle. 
Overall, it might be quite counter intuitive that adding more and more reviewers (that are also authors) to the reviewer pool does not decrease the number of papers contained in a review cycle but increases them.

\section{Conclusion}
Our work provides a first systematic analysis of \cfpr.
On the theoretical side, we show that \cfpr is a computationally hard problem even in very restricted settings,
yet practically relevant polynomial-time solvable special cases exist.
In our practical analysis, we could show that in assignments that do not care for review cycles a high fraction of authors and papers will likely be part of a short review cycle.
While collusion rings can certainly also emerge without the existence of review cycles,
for example, when authors coordinate over multiple conferences~\citep{DBLP:journals/cacm/Littman21,survey},
allowing so many easy opportunities means to leave a huge door unlocked without good reason:
Our heuristic significantly improves the situation, since it seems to always find cycle-free review assignment at a very low quality loss.

For future work, it would be valuable to further investigate the limits of our heuristic.
While our current bounds are certainly not tight, there are also clear limitations for possible extensions imposed by our NP-hardness results in quite restrictive settings from \Cref{sec:hardness}.
However, a concrete and practically very relevant open question is whether the minimum degree in our analysis can be replaced by the average degree; this would make the results much more robust against outliers.
Finally, due to the lack of data, we tested our model on just one dataset.
Obtaining more data to test our and other models on would be extremely valuable.
 
\section*{Acknowledgments}
NB was supported by the DFG project MaMu (NI
369/19) and by the DFG project ComSoc-MPMS (NI 369/22).
RB was partially supported by the DFG project AFFA (BR~5207/1 and NI~369/15).
This work was started at the research retreat of the TU Berlin
Algorithms and Computational Complexity group held in September 2020.

\bibliographystyle{plainnat}

\end{document}